\documentclass[11pt]{article}

\usepackage{amsfonts,latexsym,amsthm,amssymb,amsmath,amscd,euscript,tikz,mathtools}
\usepackage{framed}
\usepackage{authblk}
\usepackage[margin=1in]{geometry}
\usepackage{color}
\usepackage[colorlinks=true,citecolor=blue,linkcolor=blue]{hyperref}
\usepackage{hyperref}
\hypersetup{colorlinks=true,citecolor=blue,urlcolor =black,linkbordercolor={1 0 0}}
\usepackage{enumitem}
\usepackage{textcomp}
\usepackage{physics}
\usepackage{mathrsfs}
\usepackage{dsfont}
\usepackage{algorithm} %
\usepackage{algpseudocode} %
\usepackage{titlesec}
\usepackage{tikz-cd}
\usepackage{caption}
\usepackage[capitalize]{cleveref}
\RemoveFromHook{label}[firstaid/cleveref] 
\usepackage{stmaryrd}
\usepackage{graphicx}
\usepackage[normalem]{ulem}
\usetikzlibrary{positioning,chains,fit,shapes,calc}

\allowdisplaybreaks[1]

\newtheorem{theorem}{Theorem}[section]
\newtheorem{proposition}[theorem]{Proposition}
\newtheorem{lemma}[theorem]{Lemma}

\newtheorem{corollary}[theorem]{Corollary}

\newtheorem{definition}[theorem]{Definition}

\theoremstyle{remark}

\def\CC{\mathbb{C}}
\def\FF{\mathbb{F}}

\def\ZZ{\mathbb{Z}}

\def\calA{\mathcal{A}}

\def\calC{\mathcal{C}}
\def\calD{\mathcal{D}}
\def\calE{\mathcal{E}}
\def\calL{\mathcal{L}}
\def\calN{\mathcal{N}}
\def\calO{\mathcal{O}}
\def\calQ{\mathcal{Q}}

\def\frakm{\mathfrak{m}}

\newcommand{\wt}{\operatorname{wt}}
\newcommand{\dist}{\operatorname{dist}}

\newcommand{\rowsp}{\operatorname{rowsp}}

\newcommand{\CSS}{\operatorname{CSS}}
\newcommand{\punct}{\operatorname{Punct}}
\newcommand{\short}{\operatorname{Short}}
\newcommand{\supp}{\operatorname{supp}}
\newcommand{\ord}{\operatorname{ord}}
\newcommand{\Span}{\operatorname{span}}
\newcommand{\evmap}{\operatorname{ev}}

\makeatletter
\newcommand\footnoteref[1]{\protected@xdef\@thefnmark{\ref{#1}}\@footnotemark}
\makeatother

\definecolor{revisionbrown}{RGB}{160,110,0}
\newcommand{\revisioncolor}{%
  \color{revisionbrown}%
  \hypersetup{linkcolor=revisionbrown,citecolor=revisionbrown,urlcolor=revisionbrown}%
}

\title{Asymptotically Good Quantum Codes with Addressable Transversal T Gates}

\author[1]{Tongyin Lin\thanks{12232812@mail.sustech.edu.cn}}
\author[1]{Bujiao Wu\thanks{wubujiao@iqasz.cn}}
\author[2]{Bin Cheng\thanks{bincheng@nus.edu.sg}}

\affil[1]{International Quantum Academy, Shenzhen 518048, China}
\affil[2]{Centre for Quantum Technologies, National University of Singapore, 17543, Singapore}

\begin{document}


\maketitle
\thispagestyle{empty}

\begin{abstract}
Designing quantum codes with both transversal non-Clifford gates and good error-correcting
parameters is an important goal in fault-tolerant quantum computation.
Here, we construct an explicit family of binary CSS codes with
constant rate and linear distance that admit fully addressable
transversal $T$ gates.  Specifically, any prescribed
tensor product of logical powers of $T$ is implemented by a tensor
product of physical powers of $T$, without any subsequent correction.  Our construction combines algebraic-geometry codes with suitable 
binary embedding to obtain generalized divisibility of binary codes.
This divisibility then enables fully addressable transversal $T$ gates.
Furthermore, we formulate the minimum binary embedding length problem as a
minimum-weight problem over an affine space and numerically improve the constants in the asymptotic rate
and relative-distance bounds.
\end{abstract}





\section{Introduction}
Implementing logical operations while preserving protection against physical errors is a central requirement of fault-tolerant quantum computation. 
Transversal gates provide a particularly simple mechanism for achieving this goal~\cite{gottesman1998theory}. 
A transversal operation acts on disjoint sets of physical qubits, each containing at most one qubit from any given code block. 
However, the Eastin--Knill theorem establishes that a quantum code capable of detecting arbitrary single-qubit errors cannot support a universal set of transversal logical gates~\cite{eastin2009restrictions}. Universal fault-tolerant computation therefore requires additional ingredients beyond transversal gates on a fixed code.

Despite this restriction, transversal non-Clifford gates remain valuable building blocks for universal fault-tolerant architectures. Of particular importance is the gate $T:=\operatorname{diag}(1,e^{\pi i/4})$, which, together with Clifford operations, forms a universal gate set, allowing arbitrary quantum unitaries to be approximated to any desired accuracy~\cite{boykin2000new}.
Quantum error correction codes admitting a transversal logical $T$ underlie important magic-state distillation protocols, which use Clifford operations and measurements to convert noisy resource states into higher-fidelity states suitable for implementing non-Clifford gates~\cite{bravyi2005universal,bravyi2012magic}. Transversal non-Clifford gates also play a key role in code-switching and gauge-fixing approaches, where fault-tolerant transitions between codes with complementary transversal gate sets enable universal logical computation~\cite{anderson2014fault,bombin2015gauge}.

The value of transversal logical gates also depends on the encoding efficiency and error-correction capabilities of the underlying codes. An $[[N,K,D]]$ quantum code encodes $K$ logical qubits into $N$ physical qubits and has minimum distance $D$. A family of codes is asymptotically good if both $K=\Theta(N)$ and $D=\Theta(N)$. Such families combine constant encoding overhead, $N/K=O(1)$, with the ability, under ideal recovery, to correct arbitrary errors affecting up to $\lfloor(D-1)/2\rfloor=\Theta(N)$ physical qubits. They therefore provide a fundamental benchmark for protecting increasingly large quantum systems without an increasing number of physical qubits per logical qubit.
Although the existence of asymptotically good quantum codes has been known since the early development of quantum error correction~\cite{calderbank1996good,steane1996multiple}, these coding parameters alone do not guarantee useful transversal logical gates. Requiring a prescribed transversal logical action imposes additional structure on the code, and retaining both a nonvanishing rate and a nonvanishing relative distance under this requirement is a separate construction problem. For codes encoding many logical qubits, it is also desirable to select which logical qubits receive a gate. Of particular interest here is achieving asymptotic goodness together with fully addressable transversal $T$ gates: a tensor product of single-qubit physical gates can implement $T$ on any chosen logical qubit and the identity on all others, with the code and encoding fixed.



Recent constructions have demonstrated that asymptotic goodness is compatible with transversal non-Clifford logical operations. Golowich and Guruswami~\cite{golowich2025asymptotically} and Nguyen~\cite{nguyen2025good} constructed asymptotically good binary quantum codes with transversal logical $CCZ$ gates. He et al.~\cite{he2025asymptotically} extended this capability to addressable $CCZ$ gates on selected triples of logical qubits. On fixed-dimensional qudits, Wills, Hsieh, and Yamasaki~\cite{wills2025constant} used asymptotically good codes with transversal non-Clifford gates to achieve constant-overhead distillation of qubit magic states. Related developments have improved gate parallelizability, reduced the required qudit dimensions, and accelerated decoding~\cite{guemard2025good,cervia2025magic,gasnier2026quantum}. The physical non-Clifford operations underlying these results, however, act on multiple qubits or on higher-dimensional qudits. These constructions therefore do not directly provide asymptotically good binary codes in which a tensor product of single-qubit physical gates implements logical $T$ on every encoded qubit.

Achieving this single-qubit implementation with full addressability
requires favorable code parameters and control over the logical action
on individual encoded qubits.
Asymptotically good binary CSS-$T$ families satisfy the parameter
requirement, but physical transversal $T$ induces the logical
identity~\cite{berardini2025asymptotically} or the logical Clifford gate
$S^\dagger$~\cite{reddy2026asymptotically} in those constructions.
Thus, preservation of the code space alone does not ensure the desired
logical action. For constructions realizing logical $T$, Hastings and
Haah~\cite{hastings2018distillation} and Haah~\cite{haah2018towers}
obtained families with growing dimension and distance, but vanishing
rate. More recently, Wills~\cite{wills2026improved} constructed explicit
codes with transversal logical $T$, constant rate, and a distance lower
bound of $2^{\Omega(\sqrt{\log N})}$. This establishes growing distance,
but does not establish linear distance. Beyond implementing $T$
simultaneously on all encoded qubits, full addressability requires the
ability to apply $T$ to any chosen logical qubit while acting as the
identity on the others, with the code and encoding fixed. Combining
this addressability with constant rate and linear distance therefore
motivates the following question:
\begin{quote}
    \emph{Can an explicit family of asymptotically good binary quantum
    codes admit fully addressable transversal $T$ gates?}
\end{quote}

In this paper, we answer this question affirmatively.  We construct an
explicit family of binary Calderbank--Shor--Steane (CSS) codes with
constant rate and linear distance that admit fully addressable transversal
logical $T$ gates.  On each code, every prescribed tensor product of
logical powers of $T$ is implemented exactly by a tensor product of
physical powers of $T$, with the code and encoding fixed independently
of the prescribed action and without any subsequent Clifford correction.
The construction combines the
multiplication properties of algebraic-geometric (AG) codes with generalized divisibility of
binary codes, using independently prescribed divisibility
coefficients to control individual logical phases.

\subsection{Main results}

We first specify the notion of transversality used in this paper.  
Write $\ZZ_8\coloneqq\ZZ/8\ZZ$ and
$T^c\coloneqq\bigotimes_{\ell=1}^K T^{c_\ell}$ for
$c\in\ZZ_8^K$.
\begin{definition}[Addressable transversal $T$ gates]
\label{def:transversal-T}
Let
$\calQ\subseteq(\CC^2)^{\otimes N}$ be an $[[N,K,D]]$ quantum
code, and let
$E:(\CC^2)^{\otimes K}\to(\CC^2)^{\otimes N}$ be an encoding
isometry with image $\calQ$.  For $c\in\ZZ_8^K$, we say that
$\calQ$ admits a transversal implementation of logical $T^c$ if
there exist one-qubit physical unitaries $U_1(c),\ldots,U_N(c)$ such that
\begin{equation}
    \left(\bigotimes_{j=1}^N U_j(c)\right)E=ET^c.
\end{equation}
We say that $\calQ$ admits \emph{fully addressable transversal
$T$ gates} if this holds for every $c\in\ZZ_8^K$.
\end{definition}

The choice $c=e_\ell$ implements $T$ on the $\ell$-th logical qubit
and the identity on all other logical qubits.  Conversely, composing
and taking powers of these individual implementations gives every
$T^c$.  Our main result establishes asymptotic goodness under this
definition, with each physical tensor factor a power of $T$.

\begin{theorem}[Informal version of \cref{thm:direct-good-transversal-T}]
\label{thm:intro-main}
There is an explicit family of binary CSS codes with parameters
$[[N_i,K_i=\Theta(N_i),D_i=\Theta(N_i)]]$ and admitting fully addressable transversal logical $T$ gates.
More specifically, for every $c\in\ZZ_8^{K_i}$, the logical $T^c$ gate can be implemented by powers of physical $T$ gates.
\end{theorem}

The physical exponents depend linearly on the prescribed logical
coefficients over $\ZZ_8$.  They may be zero or even for a general
$c$; for $c=\vb{1}_{K_i}$, they are all odd.  The logical action in
\cref{thm:intro-main} is exact and requires no subsequent
Clifford correction.
The numerical specialization of this explicit family is
recorded in \cref{cor:multiplication-family-parameters} of
Appendix~\ref{app:explicit-parameters}.

The classical ingredient is a family of asymptotically good binary
codes with a generalized divisibility property.  Haah's formulation
uses odd coefficient vectors~\cite{haah2018towers}.  Here, we also allow
even coefficients and consider the full $\ZZ_8$-module
\begin{equation}
    \mathcal N_8(\calC)
    \coloneqq\left\{a\in\ZZ_8^n:
       \sum_{j=1}^n a_jx_j\equiv0\pmod8
       \text{ for every }x\in\calC\right\}.
\end{equation}
If coefficients in this module can be prescribed arbitrarily on a set
of information coordinates, puncturing and shortening on those
coordinates produce fully addressable transversal $T$ gates
(\cref{lem:weighted-divisibility-transversal-T}).

{By \cref{lem:parity-coefficient-lift}, every vector in
$(\calC^{\star4})^\perp$ admits a weighted $8$-divisibility lift
with the same parity.  For a code with the $4$-multiplication property,
this applies to every codeword.  Lifting suitable codewords gives a
coefficient matrix whose restriction to the selected information
coordinates is invertible over $\ZZ_8$.  Normalizing this restriction
gives the prescribed coefficients required by the addressability
criterion (\cref{lem:prescribed-coefficients}).}

{We obtain the required binary codes from AG codes over
$\FF_{1024}$ with the $15$-multiplication property and positive rate
and primal and dual relative distances
(\cref{lem:good-ag-fifteen-multiplication-code}).  The fixed binary
embedding in \cref{lem:good-weighted-divisible-code} converts their
sixteenfold product identity into vanishing fivefold binary overlaps.
Restricting selected field coordinates to a one-dimensional binary
subspace generated by an element of trace one and compressing their
images preserves the $4$-multiplication property.  Applying the
divisibility criterion to the compressed code gives fully addressable
transversal $T$ gates.  The outer primal and dual distances yield
linear lower bounds on the quantum $X$- and $Z$-distances, respectively
(\cref{thm:direct-good-transversal-T}).}

{\Cref{sec:minimum-length-embeddings} expresses the minimum
length of an embedding satisfying the fivefold identity as the minimum
Hamming weight in a binary affine space.  We give methods to find and
verify shorter embeddings.  For the $\FF_{1024}$ construction, these
improve the asymptotic rate and the relative minimum-distance lower
bound.  Explicit supports, certified inner parameters, and the resulting
quantum-code bounds are recorded in \cref{app:explicit-parameters}.}

\subsection{Related work}

{Bravyi and Haah~\cite{bravyi2012magic} introduced triorthogonal
matrices to construct CSS codes for magic-state distillation. Their
conditions on the overlaps of distinct pairs and triples of rows ensure
that physical transversal $T$, possibly followed by a Clifford
correction, implements $T$ on every logical qubit.
Hastings and Haah~\cite{hastings2018distillation} obtained sublogarithmic
distillation overhead using punctured Reed--Muller codes.
Haah and Hastings~\cite{haah2018codes} generalized triorthogonality to
construct distillation protocols for $T$, controlled-$S$, and Toffoli
gates. Rengaswamy et al.~\cite{rengaswamy2020optimality} characterized
the conditions under which physical transversal $T$ on a CSS code
implements $T$ on every logical qubit.}

{The connection between codeword weights and overlap conditions
is formalized by Ward's divisibility criterion, obtained by polarizing
the weight function~\cite{ward1990weight}.
Haah~\cite{haah2018towers} introduced generalized divisibility with odd
coefficient vectors and constructed towers of codes supporting diagonal
gates at successive levels of the Clifford hierarchy.
For odd coefficient vectors, our weighted $8$-divisibility agrees with
his level-three divisibility condition. We allow coefficients throughout
$\ZZ_8$, including even and zero entries, and use the binary
$4$-multiplication property to prescribe their values on selected
information coordinates (\cref{lem:prescribed-coefficients}).
Together with \cref{lem:weighted-divisibility-transversal-T}, this yields
independently prescribed logical powers of $T$ with a fixed encoding
and no subsequent Clifford correction.}

{AG codes provide a route to combining transversal gates with
good asymptotic parameters. Building on the Reed-Solomon approach of
Krishna and Tillich~\cite{krishna2019towards}, Golowich and
Guruswami~\cite{golowich2025asymptotically} and
Nguyen~\cite{nguyen2025good} combined AG codes with alphabet reduction
to obtain asymptotically good binary codes with transversal $CCZ$.
He et al.~\cite{he2025asymptotically} obtained fully addressable $CCZ$
on asymptotically good binary codes; their earlier addressable-orthogonality
framework also treats logical $T$ up to Clifford
corrections~\cite{he2025quantum}.
Extensions to parallelizable multi-controlled gates appear
in~\cite{guemard2025good,gasnier2026quantum}.
AG codes also underlie the constant-overhead magic-state distillation
protocol of Wills, Hsieh, and Yamasaki~\cite{wills2025constant}.}

{For physical transversal $T$, the induced logical action is
a separate requirement from preservation of the code space.
Berardini et al.~\cite{berardini2025asymptotically} and Reddy and
Kashyap~\cite{reddy2026asymptotically} constructed asymptotically good
CSS-T codes on which this operation induces the logical identity and
logical $S^\dagger$ on every encoded qubit, respectively.
General descriptions of diagonal logical actions are given
in~\cite{hu2022designing,camps2026transversal}, and Reddy and
Kashyap~\cite{reddy2026realizing} developed an appending construction
for prescribed logical rotations, including addressable ones.
Further CSS-T constructions are studied in~\cite{jaramillo2026matrix}.}

{Other constructions target logical $T$ directly.
Wills~\cite{wills2026improved} used divisible monomial codes to obtain
constant-rate families with growing distance and simultaneous logical
$T$ without a Clifford correction. Cao and
Lackey~\cite{cao2026quantum} developed tensor-network constructions
with transversal and addressable gates.
Short-code constructions and decoding methods are studied
in~\cite{jain2025transversal,baldelli2026constructing}; related
code-switching approaches appear
in~\cite{anderson2014fault,dastbasteh2026quantum}.}

{In independent and concurrent work, San-Jos\'e constructed
asymptotically good binary CSS codes with transversal diagonal rotations
at every fixed level of the Clifford hierarchy from the third onward,
including constructions requiring no subsequent correction.
He also obtained addressable rotations by applying a hierarchy-lowering
argument to the construction at the next
level~\cite[Remark~6.8]{sanjose2026asymptotically}.
We instead obtain fully addressable transversal $T$ directly by prescribing
weighted $8$-divisibility coefficients under a fivefold binary overlap
condition, without invoking a construction at a higher level of the
Clifford hierarchy
(\cref{lem:prescribed-coefficients,lem:weighted-divisibility-transversal-T}).}

\section{Preliminaries}

\paragraph{Notations.}
Let $[n]\coloneqq\{1,\ldots,n\}$.
We write $\overline{S}$ for the complement of $S$ and $e_\ell$ for the standard unit vector at coordinate $\ell$, with the index sets understood from context.
For $q=2^m$, let $\Tr:\FF_q\to\FF_2$ denote the absolute trace, given by $\Tr(x)\coloneqq\sum_{j=0}^{m-1}x^{2^j}$~\cite{lidl1997finite}.
We also write $\Tr_m$ to specify the extension degree.

\paragraph{Classical codes.}

A vector is a row vector unless otherwise specified.
A classical code $\calC \subseteq \FF_q^n$ is a linear subspace of $\FF_q^n$.
The dual code $\calC^\perp$ is the orthogonal complement of $\calC$ with respect to the standard inner product.
For $v\in\FF_q^n$, its Hamming weight $\wt(v)$ is the number of its nonzero entries; we may also use $|v|$ to denote its Hamming weight.
The distance of a nonzero linear code $\calC$ is $\dist(\calC)\coloneqq\min\{\wt(c):c\in\calC\setminus\{0\}\}$.
For a matrix $M$, let $\rowsp(M)$ denote the linear span of its rows.
For $I\subseteq[n]$, let
$c_I\coloneqq(c_j)_{j\in I}$.  The puncturing of $\calC$ on $I$ is
\begin{equation}
    \punct_I(\calC)
    \coloneqq\{c_{[n]\setminus I}:c\in\calC\}.
\end{equation}
The shortening of $\calC$ on $I$ is
\begin{equation}
    \short_I(\calC)
    \coloneqq
    \left\{c_{[n]\setminus I}:
    c\in\calC,\ c_j=0\ \text{for every }j\in I\right\}.
\end{equation}

We will also use the multiplication property of classical codes, defined as follows.
\begin{definition}[$t$-multiplication property]
\label{def:t-multiplication-property}
\label{def:star-product}
Let $t\geq1$ be an integer.
A linear code $\calC\subseteq\FF_q^n$ has the \emph{$t$-multiplication property}
if for any codewords $c^{(1)}, \ldots, c^{(t+1)} \in \calC$, their component-wise product satisfies
\begin{equation}
    \left| c^{(1)} \star \cdots \star c^{(t+1)} \right| := \sum_{j=1}^n c^{(1)}_j \cdots c_{j}^{(t+1)} \equiv 0 \pmod q.
\end{equation}
Equivalently, the code $\calC$ satisfies the $t$-multiplication property if and only if $\calC^{\star t}\subseteq\calC^\perp$, where $\calC^{\star t}
\coloneqq
\Span_{\FF_q}
\left\{
c^{(1)}\star\cdots\star c^{(t)}:
c^{(1)},\ldots,c^{(t)}\in\calC
\right\}.$
\end{definition}
More generally, a multiplication property requires
$\calC^{\star t}\subseteq\calD$ for a prescribed target
code $\calD$~\cite{couvreur2021algebraic}.  Our choice of
$\calD=\calC^\perp$ follows Nguyen~\cite{nguyen2025good}.

\paragraph{CSS codes.}
Let $\calC_0\subseteq\calC_1\subseteq\FF_2^N$ be binary linear codes.  For
$a,b\in\FF_2^N$, define
$X(a)\coloneqq\bigotimes_{i=1}^N X^{a_i}$ and
$Z(b)\coloneqq\bigotimes_{i=1}^N Z^{b_i}$.  The CSS code
$\CSS(\calC_0,\calC_1)$ is the joint $+1$ eigenspace of the stabilizer group $\left\langle X(a) : a\in \calC_0,\; Z(b) : b\in \calC_1^\perp \right\rangle$
or, equivalently, the span of the coset states
\begin{equation}
    \ket{c+\calC_0}
    \coloneqq \frac{1}{\sqrt{|\calC_0|}}
    \sum_{a\in\calC_0}\ket{c+a},
    \qquad c\in\calC_1.
\end{equation}
The inclusion $\calC_0\subseteq\calC_1$ ensures that the stabilizer
generators commute.  Its rank is
$\dim\calC_0+\dim\calC_1^\perp
=N-\dim\calC_1+\dim\calC_0$, so the code encodes
$K=\dim\calC_1-\dim\calC_0$ logical qubits.
The $X$- and $Z$-distances of $\CSS(\calC_0,\calC_1)$ are,
respectively,
\begin{align}
    D_X\coloneqq \min\left\{\wt(x) :
            x\in \calC_1\setminus\calC_0\right\}, \qquad
    D_Z\coloneqq \min\left\{\wt(z) :
            z\in \calC_0^\perp\setminus\calC_1^\perp\right\}.
\end{align}
Equivalently, $D_X$ and $D_Z$ are the minimum weights of nontrivial logical
$X$- and $Z$-type Pauli operators.  The distance of the quantum code is
$D=\min\{D_X,D_Z\}$.

A CSS code is \emph{asymptotically good} if, for sufficiently large block length $N$, its dimension and
both directional distances grow linearly with the $N$; namely, $K=\Omega(N)$, $D_X=\Omega(N)$, and $D_Z=\Omega(N)$.

\paragraph{Divisibility.}
A code $\calC$ is divisible by an integer $\eta$ if every codeword has Hamming weight
divisible by $\eta$~\cite{ward1990weight}.
Divisibility turns out to be closely related to transversality.
We use a weighted version of this property with arbitrary
coefficients modulo $8$.
\begin{definition}[Weighted $8$-divisibility]
\label{def:weighted-8-divisibility}
Let $a=(a_1,\ldots,a_n)\in\ZZ_8^n$ be a coefficient vector.
Define the $a$-weighted Hamming weight of $c\in\FF_2^n$ by
$\wt_a(c)\coloneqq\sum_{j=1}^n a_jc_j$.
A binary linear code $\calC\subseteq\FF_2^n$ is
\emph{$a$-weighted $8$-divisible} if $\wt_a(c)\equiv0\pmod8$ for
every $c\in\calC$.
For $a=\vb{1}_n$ (all entries of $a$ are 1), this reduces to ordinary $8$-divisibility.
\end{definition}
Haah~\cite{haah2018towers} studied weighted divisibility with odd
coefficient vectors in the construction of CSS codes with transversal
gates at successive levels of the Clifford hierarchy.
For a fixed code $\calC$, the coefficient vectors giving weighted
$8$-divisibility form the $\ZZ_8$-module,
\begin{equation}
    \mathcal N_8(\calC)
    \coloneqq\{a\in\ZZ_8^n:
       \wt_a(x)\equiv0\pmod8\text{ for every }x\in\calC\}.
    \label{eq:divisibility-coefficient-module}
\end{equation}

\section{Addressable transversal \texorpdfstring{$T$}{T} gates from generalized divisibility}

In this section, we show that weighted divisibility implies fully addressable transversality of logical $T$.
The construction of CSS codes uses the puncturing-and-shortening construction of Krishna and
Tillich~\cite{krishna2019towards}, which is also used in
\cite{nguyen2025good}.

\begin{lemma}
\label{lem:weighted-divisibility-transversal-T}
Let $\calC\subseteq\FF_2^n$ be a binary $[n,k,d]$ code, and let
$S\subseteq[n]$ be a set of $K\geq1$ coordinates such that
$\{c_S: c \in \calC\} =\FF_2^S$.
Suppose that, for every $\ell\in S$, there is a vector
$b^{(\ell)}\in\ZZ_8^{\overline S}$ such that $\calC$ is
$(-e_\ell,b^{(\ell)})$-weighted $8$-divisible.
Then, the codes $\calC_0\coloneqq\short_S(\calC)$ and $\calC_1\coloneqq\punct_S(\calC)$
define a CSS code $\CSS(\calC_0, \calC_1)$ with parameters $[[N=n-K,K,D]]$ and a fixed
encoding admitting fully addressable transversal $T$ gates implemented
by physical powers of $T$.
Moreover, $D_X\geq d-K$, $D_Z=\min\{\wt(z):z\in\calC_0^\perp\setminus\calC_1^\perp\}$ and $D = \min\{D_X,D_Z\}$.
\end{lemma}

\begin{proof}
Since
$\{c_S: c \in \calC\}=\FF_2^S$, we may choose a full-row-rank
generator matrix of $\calC$ of the form
\begin{equation}
    G=\begin{pmatrix}I_K&H_1\\0&H_0\end{pmatrix}.
    \label{eq:addressable-generator}
\end{equation}
We next show that the rows of $\binom{H_1}{H_0}$ are linearly
independent.  Suppose that $xH_1+yH_0=0$ for
$x\in\FF_2^K$ and $y\in\FF_2^{k-K}$.  Then
\begin{equation}
    (x,0)=x(I_K,H_1)+y(0,H_0)\in\calC.
\end{equation}
For every $\ell\in S$, weighted divisibility gives
\begin{equation}
    0\equiv\wt_{(-e_\ell,b^{(\ell)})}(x,0)=-x_\ell\pmod8.
\end{equation}
Since $x_\ell\in\{0,1\}$, we obtain $x=0$ and hence $yH_0=0$.
The rows of $H_0$ are linearly independent because $G$ has full row
rank, so $y=0$.  This proves the claimed linear independence.
Consequently,
\begin{equation}
    \calC_0=\rowsp(H_0),\qquad
    \calC_1=\rowsp\begin{pmatrix}H_1\\H_0\end{pmatrix},
    \qquad
    \dim\calC_1-\dim\calC_0=K.
\end{equation}
Fix the encoding
\begin{equation}
    E\ket{x}=\frac1{\sqrt{|\calC_0|}}
       \sum_{v\in\calC_0}\ket{xH_1+v},
    \qquad x\in\FF_2^K.
    \label{eq:addressable-encoding}
\end{equation}
The cosets in this expression are distinct, so these states are
orthonormal and span $\CSS(\calC_0,\calC_1)$.

We now prove that this fixed encoding admits fully addressable
transversal $T$ gates.
For any $c\in\ZZ_8^S$, we have $E T^c \ket{x}=\omega_8^{\sum_\ell c_\ell x_\ell}E\ket{x}$, where $\omega_8=e^{\pi i/4}$.
Set
\begin{equation}
    b_c\coloneqq\sum_{\ell\in S}c_\ell b^{(\ell)}
       \in\ZZ_8^{\overline S}.
\end{equation}
Taking the same linear combination of the assumed divisibility
relations shows that $\calC$ is $(-c,b_c)$-weighted $8$-divisible.
For every $x\in\FF_2^S$ and $v\in\calC_0$, the vector
$(x,xH_1+v)$ belongs to $\calC$, so
\begin{equation}
    \wt_{b_c}(xH_1+v)
       \equiv\sum_{\ell\in S}c_\ell x_\ell\pmod8.
    \label{eq:exact-addressable-phase}
\end{equation}
Note that the action of $\left(\bigotimes_{j\in\overline S}T^{(b_c)_j}\right)$ on every summand in
\cref{eq:addressable-encoding} gives a phase $\omega_8^{\wt_{b_c}(xH_1+v)}$.  
Hence,
\begin{equation}
    \left(\bigotimes_{j\in\overline S}T^{(b_c)_j}\right)E=ET^c
    \qquad(c\in\ZZ_8^S).
    \label{eq:addressable-puncturing-action}
\end{equation}
This proves fully addressable transversality with the encoding $E$
fixed independently of $c$.

Finally, if $w=xH_1+v\in\calC_1\setminus\calC_0$, then $x\neq0$
and $(x,w)\in\calC\setminus\{0\}$.  So, $\wt(w)\geq d-\wt(x)\geq d-K$.
This proves the $X$-distance bound.
The $Z$-distance bound is by definition.
\end{proof}

\section{Good classical codes with generalized divisibility}
\label{sec:gen_divisibility}

We construct asymptotically good binary codes whose divisibility
coefficients can be prescribed on any set of information coordinates.
First, \cref{lem:prescribed-coefficients} derives the weighted
$8$-divisibility conditions required in
\cref{lem:weighted-divisibility-transversal-T} from the binary
$4$-multiplication property, using a parity-lifting argument in its
proof.  We then construct AG codes over $\FF_{1024}$ with the
$15$-multiplication property and positive rate and positive primal and
dual relative distances in
\cref{lem:good-ag-fifteen-multiplication-code}.
\Cref{lem:good-weighted-divisible-code} maps these codes to binary
codes with the $4$-multiplication property.  Finally,
\cref{thm:direct-good-transversal-T} combines this construction with
\cref{lem:weighted-divisibility-transversal-T} to obtain asymptotically
good CSS codes with fully addressable transversal $T$ gates.
For the AG-code properties used in this section, we refer to the standard textbook by
Stichtenoth~\cite{stichtenoth2009algebraic}; for completeness, we also provide a
self-contained overview in
Appendix~\ref{app:ag-codes}.

\subsection{Generalized divisibility from the multiplication property}

We now show that the binary $4$-multiplication property implies the
weighted-divisibility conditions required in
\cref{lem:weighted-divisibility-transversal-T}.  The main ingredient
is a parity-lifting argument that converts a vector in
$(\calC^{\star4})^\perp$ into a weighted $8$-divisibility coefficient
vector with the prescribed parity.

\begin{lemma}
\label{lem:prescribed-coefficients}
Let $\calC\subseteq\FF_2^n$ be a code satisfying the $4$-multiplication property,
and let $S\subseteq[n]$ be a set of $K\geq1$ coordinates such that
$\{c_S:c\in\calC\}=\FF_2^S$.
Then, for every $\ell\in S$, there exists an efficiently computable vector
$b^{(\ell)}\in\ZZ_8^{\overline S}$ such that $\calC$ is
$(-e_\ell,b^{(\ell)})$-weighted $8$-divisible.
\end{lemma}

Hence, the hypotheses of
\cref{lem:weighted-divisibility-transversal-T} are satisfied.
The proof of \cref{lem:prescribed-coefficients} relies on the following parity-lifting lemma.

\begin{lemma}
\label{lem:parity-coefficient-lift}
Let $\calC\subseteq\FF_2^n$ be a binary linear code.
For every
$p\in(\calC^{\star4})^\perp$, there exists an efficiently computable vector
$a\in\calN_8(\calC)$ such that $a\equiv p\pmod2$.
\end{lemma}

\begin{proof}
Every vector $a\in\ZZ_8^n$ with $a\equiv p\pmod2$ can be uniquely
written as $a=p+2t+4z$, where $t,z\in\FF_2^n$.
We construct $t$ and $z$ successively so that
$a\in\calN_8(\calC)$.

Set $\calE\coloneqq\calC^{\star2}$.
For $u,v\in\calE$, the vector $u\star v$ lies in
$\calC^{\star4}$.
Since $p\in(\calC^{\star4})^\perp$, we have $\sum_{j=1}^n p_j u_jv_j\equiv0\pmod2$.
Taking $v=u$ shows that $\wt_p(u)$ is even.
Define $\varphi_p(u)\coloneqq\wt_p(u)/2\pmod2$.
For any $u,v\in\calE$,
\begin{equation*}
    \frac{\wt_p(u+v)}{2}
    =\frac{\wt_p(u)}{2}+\frac{\wt_p(v)}{2}
    -\sum_{j=1}^n p_j u_jv_j.
\end{equation*}
The last term vanishes modulo $2$, so $\varphi_p$ is linear on $\calE$.

Choose a basis
$e^{(1)},\ldots,e^{(k_{\calE})}$ of $\calE$ and solve
$e^{(\ell)}\cdot t=\varphi_p(e^{(\ell)})$ for $\ell\in[k_{\calE}]$.
The equations are linearly independent, so a solution exists.
By linearity, $t\cdot u=\varphi_p(u)$ for every $u\in\calE$.

Define $a^{(1)}\coloneqq p+2t$.
Then, for every $u\in\calE$,
\begin{equation*}
    \begin{aligned}
        \wt_{a^{(1)}}(u)
        &=\wt_p(u)+2\sum_j t_j u_j \\
        &\equiv2\wt_p(u)\equiv0\pmod4.
    \end{aligned}
\end{equation*}
In particular, taking $u=x\star y$ with $x,y\in\calC$ gives
$\sum_{j=1}^n a_j^{(1)}x_jy_j\equiv0\pmod4$.
Taking $y=x$ shows that
$\wt_{a^{(1)}}(x)$ is divisible by $4$ for every $x\in\calC$.
Define $\psi(x)\coloneqq\wt_{a^{(1)}}(x)/4\pmod2$ for $x\in\calC$.
For $x,y\in\calC$,
\begin{equation*}
    \frac{\wt_{a^{(1)}}(x+y)}4
    =\frac{\wt_{a^{(1)}}(x)}4+\frac{\wt_{a^{(1)}}(y)}4
    -\frac12\sum_{j=1}^n a_j^{(1)}x_jy_j\pmod2.
\end{equation*}
The last term vanishes modulo $2$, so $\psi$ is linear over $\FF_2$.

Choose a basis
$g^{(1)},\ldots,g^{(k)}$ of $\calC$ and solve
$g^{(\ell)}\cdot z=\psi(g^{(\ell)})$ for $\ell\in[k]$.
Again, a solution exists, and by linearity, $z\cdot x=\psi(x)$ for every $x\in\calC$.
Thus, the vector $a=p+2t+4z$ satisfies
\begin{equation*}
    \wt_a(x)\equiv\wt_{a^{(1)}}(x)+4\psi(x)\equiv0\pmod8
    \qquad\text{for every }x\in\calC.
\end{equation*}
Hence,
$a\in\calN_8(\calC)$ and
$a\equiv p\pmod2$.

A basis of $\calE$ can be extracted from the pairwise products
$g^{(i)}\star g^{(j)}$ by Gaussian elimination.  Both linear systems
have at most $n$ independent equations in $n$ variables, so the
construction has polynomial complexity.
\end{proof}

\begin{proof}[Proof of \cref{lem:prescribed-coefficients}]
For each $\ell\in S$, choose
$p^{(\ell)}\in\calC$ such that $p^{(\ell)}_S=e_\ell$.
Since $\calC$ has the $4$-multiplication property, $\calC^{\star4}\subseteq\calC^\perp$, and $p^{(\ell)}\in\calC$, which implies that $p^{(\ell)}\in(\calC^{\star4})^\perp$.
By \cref{lem:parity-coefficient-lift}, there exists
$a^{(\ell)}\in\calN_8(\calC)$ such that $a^{(\ell)}\equiv p^{(\ell)}\pmod2$.
In particular, $a^{(\ell)}_S\equiv e_\ell\pmod2$.

Let $A\in\ZZ_8^{K\times n}$ have the vectors $a^{(\ell)}$ as its
rows, indexed by $\ell\in S$.  Then, $A_S\equiv I_K\pmod2$.
Hence, $\det(A_S)\equiv1\pmod2$, which implies that $\det(A_S)$ is odd and $A_S$ is invertible over $\ZZ_8$.

Let $M\coloneqq A_S^{-1}A$.
Then, $M_S=A_S^{-1}A_S=I_K$.
For each $\ell\in S$, denote the $\ell$-th row of $M$ by
$m^{(\ell)}$.
Thus, $m^{(\ell)}_S=e_\ell$.

For every $c\in\calC$,
\begin{align*}
    \wt_{m^{(\ell)}}(c) &=\sum_{j=1}^n m_j^{(\ell)}c_j 
    =\sum_{i\in S}(A_S^{-1})_{\ell i}
      \sum_{j=1}^n a_j^{(i)}c_j \\
    &=\sum_{i\in S}(A_S^{-1})_{\ell i}
      \wt_{a^{(i)}}(c) \\
    &\equiv 0\pmod8,
\end{align*}
where the last congruence follows from
$a^{(i)}\in\calN_8(\calC)$.

Multiplying all the coefficients in $m^{(\ell)}$ by $-1$ preserves
divisibility by $8$, since $\wt_{-m^{(\ell)}}(c)=-\wt_{m^{(\ell)}}(c)\equiv0\pmod8$
for every $c\in\calC$.

Since $m^{(\ell)}_S=e_\ell$, write $m^{(\ell)}=\bigl(e_\ell,b'^{(\ell)}\bigr)$
for some $b'^{(\ell)}\in\ZZ_8^{\overline{S}}$.
Then, $-m^{(\ell)}=\bigl(-e_\ell,-b'^{(\ell)}\bigr)\in\calN_8(\calC)$.
Setting $b^{(\ell)}\coloneqq-b'^{(\ell)}$,
we conclude that $\calC$ is
$(-e_\ell,b^{(\ell)})$-weighted $8$-divisible for every
$\ell\in S$.

The required vectors are computable in polynomial time by the two
binary linear systems in the proof of \cref{lem:parity-coefficient-lift} and the
subsequent matrix operations.
\end{proof}

\subsection{Algebraic-geometry codes with the \texorpdfstring{$15$}{15}-multiplication property}
\label{subsec:ag-fifteen-multiplication}

We next construct codes over $\FF_{1024}$ with the
$15$-multiplication property.  The construction uses the same
Galois tower and canonical divisor that underlie the
self-dual-code construction in
\cite[Theorem~8.4.9]{stichtenoth2009algebraic} and the
multiplication-property construction in
\cite[Theorem~3.6]{nguyen2025good}, with a divisor chosen for a
sixteen-product identity.  It is also the $q=1024$, $\tau=16$
instance of the AG construction in
\cite{sanjose2026asymptotically}.
Appendix~\ref{app:ag-codes} collects the function-field definitions,
parameter bounds, and multiplication criterion used below.

\begin{lemma}
\label{lem:good-ag-fifteen-multiplication-code}
There is an explicit sequence of $[n_{\calA_i},k_{\calA_i},\dist(\calA_i)]$ AG codes $\{\calA_i\}$ over $\FF_{1024}$.
Each $\calA_i$ satisfies the $15$-multiplication property,
$\calA_i^{\star15}\subseteq\calA_i^\perp$.
For all sufficiently large $i$, the parameters satisfy
\begin{equation}
    \label{eq:direct-outer-limits}
    \frac{k_{\calA_i}}{n_{\calA_i}}
       \longrightarrow \frac{17}{496},
    \qquad
    \frac{\dist(\calA_i)}{n_{\calA_i}}
       \geq\frac{463}{496},
    \qquad
    \frac{\dist(\calA_i^\perp)}{n_{\calA_i}}
       \geq\frac1{496}.
\end{equation}
\end{lemma}

\begin{proof}
We use the standard form of the optimal Galois tower over
$\FF_{1024}=\FF_{32^2}$~\cite{garcia1995tower}.
At tower level $i$, let the function field
$\mathcal F_i/\FF_{1024}$ have genus $g_i$, and let $D_i$ be the sum
of all $n_{\calA_i}$ distinct rational places.
The number of rational places grows with $i$, $n_{\calA_i} = 31^{O(i)}$.
There are positive
divisors $A_i,B_i$, disjoint from $D_i$, and even integers
$\alpha_i,\beta_i\to\infty$ such that
$\alpha_i\deg A_i=\beta_i\deg B_i=n_{\calA_i}/31$ and
\begin{equation}
    g_i=1+\frac{n_{\calA_i}}{31}
       \left(1-\frac1{\alpha_i}-\frac1{\beta_i}\right).
    \label{eq:direct-tower-genus}
\end{equation}
In particular, $g_i/n_{\calA_i}\to1/31$.
There is also a suitable canonical divisor
$R_i=(32\alpha_i-2)A_i+(\beta_i-2)B_i-D_i$.
Suppress the index $i$ and define $G=(2\alpha-1)A+\lfloor(\beta-2)/16\rfloor B$.
Then, $16G\leq R+D$, because $16(2\alpha-1)\leq32\alpha-2$ and
$16\lfloor(\beta-2)/16\rfloor\leq\beta-2$.

For the Riemann-Roch space
$\calL(G)=\{f\in\mathcal F_i^\times:(f)+G\geq0\}\cup\{0\}$, define the AG code as
\begin{equation}
    \calA(G)\coloneqq
    \left\{(f(P))_{P\in\supp D}:
       f\in\calL(G)\right\}
    \subseteq\FF_{1024}^{n_{\calA}}.
    \label{eq:direct-outer-code}
\end{equation}
We simply write $\calA$ when the context is clear.
The inequality $16G\leq R+D$ gives $15G\leq R+D-G$.
By the star-product and inclusion properties of
AG codes~\cite[Fact~3.4]{nguyen2025good},
\begin{equation*}
    \calA(G)^{\star 15}
       \subseteq\calA(15G)
       \subseteq\calA(R+D-G).
\end{equation*}
By AG-code duality~\cite[Theorem~2.2.7 and Proposition~2.2.10]
{stichtenoth2009algebraic}, the latter code is the dual of
$\calA(G)$, namely $\calA(R+D-G)=\calA(G)^\perp$.
Thus, $\calA(G)^{\star15}\subseteq\calA(G)^\perp$, which is exactly the $15$-multiplication property.
Equivalently, for
any $u^{(1)},\ldots,u^{(16)}\in\calA$,
\begin{equation}
    \sum_{P\in\supp D}
       \prod_{a=1}^{16}u_P^{(a)}=0.
    \label{eq:sixteen-product-identity}
\end{equation}

It remains to prove the asymptotic goodness of the parameters.
Set $\Delta\coloneqq\deg G+2-2g$.
The formulas for $g$ and $G$ give
\begin{align}
    \frac{\deg G}{n_{\calA}}
       &=\frac{2\alpha-1}{31\alpha}
         +\frac1{31\beta}
            \left\lfloor\frac{\beta-2}{16}\right\rfloor
       \leq\frac{33}{496},
       \label{eq:direct-degree-limit}\\
    \frac{\Delta}{n_{\calA}}
       &=\frac1{31\alpha}
         +\frac{\lfloor(\beta-2)/16\rfloor+2}{31\beta}
       \geq\frac1{496}.
       \label{eq:direct-dual-relative-bound}
\end{align}
Here, the inequalities follow from $x-1\leq\lfloor x\rfloor\leq x$.
As $\alpha,\beta\to\infty$, these ratios converge to $33/496$ and $1/496$, respectively.
Together with $g/n_{\calA}\to1/31$, this gives $2g-1\leq\deg G<n_{\calA}$ for all sufficiently large tower levels.
For these levels, the standard AG-code bounds give
$k_{\calA}\coloneqq\dim_{\FF_{1024}}\calA=\deg G+1-g$,
$\dist(\calA)\geq n_{\calA}-\deg G$, and
$\dist(\calA^\perp)\geq\Delta$.
Consequently,
\begin{equation}
    \frac{k_{\calA}}{n_{\calA}}
       \longrightarrow\frac{17}{496},
    \qquad
    \frac{\dist(\calA)}{n_{\calA}}
       \geq\frac{463}{496},
    \qquad
    \frac{\dist(\calA^\perp)}{n_{\calA}}
       \geq\frac1{496}.
\end{equation}

The defining equations of the tower and the divisors are explicit.
Bases of the Riemann-Roch spaces can be computed by standard
algorithms~\cite{hess2002computing}.
Thus, the codes $\calA_i$ and the integers $\Delta_i$ are explicit.
\end{proof}

\subsection{From algebraic-geometry codes to binary codes}\label{sec:binary-embedding}
Now, we use the five-input specialization of the alphabet-reduction map in \cite{sanjose2026asymptotically} to convert the AG codes' $15$-multiplication property into the binary $4$-multiplication property.
Prescribing weighted $8$-divisibility coefficients then yields addressable transversal $T$ without a subsequent correction (\cref{lem:prescribed-coefficients,lem:weighted-divisibility-transversal-T}).
In contrast, the addressable construction in \cite[Remark~6.8]{sanjose2026asymptotically} uses an eight-input map through hierarchy lowering.

\begin{lemma}
\label{lem:good-weighted-divisible-code}
Let $\{\calA_i\}$ be the AG code family from
\cref{lem:good-ag-fifteen-multiplication-code}.
There are a constant $1\leq L\leq\sum_{j=1}^5\binom{10}{j}$
and an explicitly computable injective $\FF_2$-linear map
$\sigma:\FF_{1024}\to\FF_2^L$, 
such that the binary codes $\calC_i\coloneqq\sigma^{\oplus n_{\calA_i}}(\calA_i)
\subseteq\FF_2^{Ln_{\calA_i}}$ satisfy the $4$-multiplication property.
Their dimensions and distances satisfy
\begin{equation}
    \dim_{\FF_2}\calC_i=10\dim_{\FF_{1024}}\calA_i,
    \qquad
    \dist(\calC_i)\geq\dist(\calA_i).
    \label{eq:binary-code-parameters}
\end{equation}
Moreover, for all $x,y\in\FF_{1024}$,
\begin{equation}
    \sigma(x)\cdot\sigma(y)=\Tr(xy),
    \qquad
    \wt(\sigma(x))\equiv\Tr(x)\pmod2.
    \label{eq:sigma-trace-identities}
\end{equation}
\end{lemma}

Together with \cref{eq:direct-outer-limits}, these parameter relations imply that the family $\{\calC_i\}$ is asymptotically good.
Its rate converges to $85/(248L)$, and its relative distance is at least $463/(496L)$ for all sufficiently large $i$.

\begin{proof}
We first give the construction of the binary embedding map.
For $x_1,\ldots,x_5\in\FF_{1024}$, define 
\begin{equation}
    \Phi(x_1,\ldots,x_5)
       \coloneqq\Tr\!\left(
          \prod_{\substack{S\subseteq[5]\\|S|\text{ odd}}}
             \left(\sum_{j\in S}x_j\right)\right).
    \label{eq:quintic-trace-form}
\end{equation}
The map $\Phi$ is symmetric, since its definition is invariant under permutations of the five inputs.
We show that it is $\FF_2$-linear in $x_1$ in \cref{appendix:properties-of-binary-encoding}; symmetry then gives linearity in every input.
In \cref{appendix:properties-of-binary-encoding}, we also prove that the value of $\Phi$ depends only on the set of its distinct inputs,
regardless of their multiplicities.
In particular, $\Phi(x,x,x,x,y)=\Phi(x,x,x,y,y)=\Tr(xy)$ and $\Phi(x,x,x,x,x)=\Tr(x)$.

Choose a binary basis $e_1,\ldots,e_{10}$ of $\FF_{1024}$.
Every $x\in\FF_{1024}$ has a unique representation $x=\sum_{i=1}^{10}\chi_i e_i$ with $\chi_i\in\FF_2$.
For each subset $S\subseteq\{1,\ldots,10\}$, define
\begin{equation}
    \ell_S(x)\coloneqq\sum_{i\in S}\chi_i\pmod2.
    \label{eq:binary-coordinate-functional}
\end{equation}
The map $\ell_S$ is $\FF_2$-linear.
For any five basis vectors, $\prod_{a=1}^5\ell_S(e_{i_a})$ equals $1$ if $\{i_1,\ldots,i_5\}\subseteq S$, and $0$ otherwise.
Like $\Phi$, this product depends only on the set of distinct basis
vectors, regardless of their order or multiplicities.
We seek coefficients $\zeta_S\in\FF_2$ such that
\begin{equation}
    \Phi(x_1,\ldots,x_5)
       =\sum_{S\subseteq\{1,\ldots,10\}}
          \zeta_S\prod_{a=1}^5\ell_S(x_a).
    \label{eq:pure-fifth-power-decomposition}
\end{equation}
First, set $\zeta_S=0$ when $|S|=0$ or $|S|>5$.
For $S_e=\{i_1,\ldots,i_5\}$ with five distinct indices, set $\zeta_{S_e}=\Phi(e_{i_1},\ldots,e_{i_5})$.
Next, for $S_e=\{i_1,\ldots,i_4\}$ with four distinct indices, set
\begin{equation}
    \zeta_{S_e}
       =\Phi(e_{i_1},e_{i_1},e_{i_2},e_{i_3},e_{i_4})
          +\sum_{S_e\subsetneq S'}\zeta_{S'}\pmod2.
    \label{eq:quintic-subset-recursion}
\end{equation}
Apply the same construction successively when $|S_e|=3,2,1$.
This gives \cref{eq:pure-fifth-power-decomposition} on all tuples of basis vectors.
By $\FF_2$-multilinearity, the identity holds for all $x_1,\ldots,x_5\in\FF_{1024}$.

Enumerate the subsets with $\zeta_S=1$ as $S_1,\ldots,S_L$, and set $\sigma_s:=\ell_{S_s}$ for $s\in[L]$.
Define $\sigma(x)\coloneqq(\sigma_s(x))_{s=1}^L$.
Each $\sigma_s$ is $\FF_2$-linear, so $\sigma$ is also $\FF_2$-linear.
The number of coordinates satisfies $L\leq\sum_{j=1}^5\binom{10}{j}=637$.
By \cref{eq:pure-fifth-power-decomposition},
\begin{equation}
    \sum_{s=1}^L\prod_{a=1}^5\sigma_s(x_a)
       =\Phi(x_1,\ldots,x_5).
    \label{eq:quintic-embedding-identity}
\end{equation}
Taking the inputs $x,x,x,y,y$ or $x,x,x,x,x$ and using the repeated-input identities gives the two identities in \cref{eq:sigma-trace-identities}, respectively.
Finally, if $\sigma(x)=\sigma(y)$, then
$\Tr((x-y)v)=0$ for every $v\in\FF_{1024}$.
Nondegeneracy of the trace pairing implies $x=y$, so $\sigma$ is injective.

Next, we apply $\sigma$ coordinatewise to the AG codes $\calA_i$ to obtain the binary codes $\calC_i$.
Every $c\in\calC_i$ has the form $c=(\sigma(u_1),\ldots,\sigma(u_{n_{\calA_i}}))$ for some $u\in\calA_i$.
For any $c^{(1)},\ldots,c^{(5)}\in\calC_i$, let
$u^{(1)},\ldots,u^{(5)}\in\calA_i$ be the corresponding codewords.
By \cref{eq:quintic-embedding-identity}, the Hamming weight of their
fivefold component-wise product satisfies
\begin{equation}
    \left|c^{(1)}\star\cdots\star c^{(5)}\right|
       \equiv\sum_{p=1}^{n_{\calA_i}}
          \Phi(u_p^{(1)},\ldots,u_p^{(5)})\pmod2.
    \label{eq:blockwise-quintic-overlap}
\end{equation}
The $15$-multiplication property of $\calA_i$ implies that the coordinate sum of any product of sixteen codewords vanishes.
Upon expanding the product in \cref{eq:quintic-trace-form}, every
monomial inside the trace is a product of sixteen inputs, allowing
repetitions.
Hence, the coordinate sum of each monomial vanishes.
By linearity of the trace, \cref{eq:blockwise-quintic-overlap} yields $\left|c^{(1)}\star\cdots\star c^{(5)}\right|\equiv0\pmod2$.
Thus, $\calC_i$ has the $4$-multiplication property.

Finally, injectivity of $\sigma$ gives
$\dim_{\FF_2}\calC_i=10\dim_{\FF_{1024}}\calA_i$.
Every nonzero field coordinate has a nonzero binary image, so
$\dist(\calC_i)\geq\dist(\calA_i)$.
This proves \cref{eq:binary-code-parameters}.
\end{proof}

\subsection{Putting everything together}
In the previous sections, we show that there is a family of AG codes with asymptotically good parameters and the 15-multiplication property.
From these AG codes, we can explicitly construct binary codes with asymptotically good encoding rates and distances, as well as 4-multiplication property.
From \cref{lem:prescribed-coefficients}, for a binary code with the 4-multiplication property, if a suitable coordinate set $S$ exists, then the code will be weighted 8-divisible, hence giving a CSS code with addressable transversal $T$.
In this section, we show how to construct such a coordinate set $S$ and establish a family of 
asymptotically good CSS codes with fully addressable transversality.
\begin{theorem}
\label{thm:direct-good-transversal-T}
There is an explicit asymptotically good family of binary CSS codes
$\calQ_i$ with parameters $[[N_i,K_i,D_i]]$ and admitting fully transversal $T$ gate.
More precisely,  the parameters of the CSS code $\calQ_i$ satisfy
\begin{equation}
    \frac{K_i}{N_i}\longrightarrow\frac1{991L},
    \qquad
    \liminf_{i\to\infty}\frac{D_{X,i}}{N_i}
       \geq\frac{925}{991L},
    \qquad
    \liminf_{i\to\infty}\frac{D_{Z,i}}{N_i}
       \geq\frac1{991L},
    \label{eq:direct-quantum-parameters}
\end{equation}
and for a fixed encoding isometry $E_i$, any $c \in \ZZ_8^{K_i}$ and a vector $b(c) \in \ZZ_8^{N_i}$ depending on $c$, we have $T^{b(c)}E_i=E_iT^c$.
\end{theorem}

\begin{proof}
Choose a sufficiently large tower $i$ and suppress the index for simplicity.
To give the required coordinate set $S$, we first need 
an auxiliary binary linear code $\calC'$.
Take a fixed $x \in \FF_{1024}$ with $\Tr(x)=1$, and a set of coordinates $S_\calA$ of code $\calA$, whose size is $K$ and $1 \le K < \dist(\calA^\perp)$.
The auxiliary code $\calC'$ is defined as
\begin{equation}
    \calC'\coloneqq
    \left\{
    \left(h,\sigma^{\oplus(n_{\calA}-K)}
                 (u_{\scriptscriptstyle{\overline{S_{\calA}}}})\right)
    :\;
    u\in\calA,\ h\in\FF_2^K,\ u_{\scriptscriptstyle{S_{\calA}}}=xh
    \right\},
    \label{eq:auxiliary-binary-code}
\end{equation}
where $\sigma$ is the binary embedding map in \cref{lem:good-weighted-divisible-code}.
The code $\calC'$ has length $n'=K+L(n_{\calA}-K)$, dimension $k'\coloneqq\dim_{\FF_2}\calC'=10k_{\calA}-9K$, and distance $d'\coloneqq\dist(\calC')\geq\dist(\calA)$.

To see this, let $M(\calA)$ be a generator matrix of $\calA$.
For every $1\le K < \dist(\calA^\perp)$, the columns of $M(\calA)$ indexed by $S_{\calA}$ are linearly independent:
otherwise, a nontrivial dependence would give a nonzero codeword of
$\calA^\perp$ supported on $S_{\calA}$, of weight at most
$K<\dist(\calA^\perp)$.
Thus, the restriction of $\calA$ on $S_\calA$ satisfies $\left\{ u_{\scriptscriptstyle{S_\calA}}:\, u\in \calA \right\} = \FF_{1024}^{K}$.
On the other hand, since $\Tr(x)=1$, we have $x\neq0$, and
$x\FF_2^{K}$ is a binary subspace of $\FF_{1024}^{K}$, hence a subspace of $\left\{ u_{\scriptscriptstyle{S_\calA}}:\, u\in \calA \right\}$.
Among the codewords of $\calA$, the number of $u$ satisfying $u_{\scriptscriptstyle{S_\calA}}\in x\FF_2^{K}$ is $2^K {1024}^{k_\calA - K} = 2^{10k_\calA - 9K}$.
For every such $u$, there is a unique $h\in\FF_2^{K}$ such that $u_{S_{\calA}}=xh$.
We use these $u$'s and the corresponding $h$ to construct $\calC'$ as in \cref{eq:auxiliary-binary-code}.
The map $u\rightarrow \left(h,\sigma^{\oplus(n_{\calA}-K)}(u_{\scriptscriptstyle{\overline{S_{\calA}}}})\right)$ is $\FF_2$-linear and injective, proving $\calC'$ is a code space over $\FF_2$, and its claimed dimension.
The code length is $n' = K+L(n_{\calA}-K)$.
Every nonzero entry of $u_{S_{\calA}}$ is mapped to a nonzero entry of $h$, and every nonzero entry of $u_{\overline{S_{\calA}}}$ is mapped to a nonzero block.
Hence, every nonzero codeword of $\calC'$ has weight at least $\dist(\calA)$.

To verify the $4$-multiplication property, take any $c'^{(1)},\ldots,c'^{(5)}\in\calC'$.
Their corresponding vectors $h^{(i)}\in\FF_2^K$ and $u^{(i)}\in\calA$ satisfy $u_{S_{\calA}}^{(i)}=xh^{(i)}$ for $i\in[5]$.
By the parity identity in \cref{eq:sigma-trace-identities}, $\wt(\sigma(x))\equiv\Tr(x)=1\pmod2$.
Since $\sigma$ is $\FF_2$-linear, replacing each block $\sigma(xh_j^{(i)})$ by the bit $h_j^{(i)}$ preserves the parity of the fivefold overlap.
Therefore,
\begin{equation*}
    \begin{aligned}
        \left|c'^{(1)}\star\cdots\star c'^{(5)}\right|
        &\equiv\sum_{j=1}^K\prod_{a=1}^5 h_j^{(a)}
            +\sum_{p\in\overline{S_{\calA}}}\sum_{s=1}^L
                \prod_{a=1}^5\sigma_s(u_p^{(a)})\\
        &\equiv\sum_{p=1}^{n_{\calA}}\sum_{s=1}^L
                \prod_{a=1}^5\sigma_s(u_p^{(a)})\equiv0\pmod2,
    \end{aligned}
\end{equation*}
where the last congruence follows from the $4$-multiplication property
of $\calC=\sigma^{\oplus n_{\calA}}(\calA)$ in
\cref{lem:good-weighted-divisible-code}.

The restriction of $\calC'$ to its first $K$ coordinates is $\FF_2^K$, so these coordinates form the required set $S$.
By \cref{lem:prescribed-coefficients,lem:weighted-divisibility-transversal-T}, the CSS code $\CSS(\calC_0,\calC_1)$ with $\calC_0\coloneqq\short_S(\calC')$ and $\calC_1\coloneqq\punct_S(\calC')$ admits a transversal implementation of logical $T$ on any set of logical qubits.
The code has parameters $[[N=L(n_{\calA}-K),K,D=\min\{D_X,D_Z\}]]$.
By \cref{lem:weighted-divisibility-transversal-T}, its $X$-distance satisfies
$D_X\geq\dist(\calC')-K$, so, together with $\dist(\calC')\geq\dist(\calA)$ proved above, we have $D_X\geq\dist(\calA)-K$.
It remains to bound $D_Z$ and choose $K$ so that both relative distances and the rate stay bounded away from zero.

We claim that
\begin{equation}
    D_Z\geq\dist(\calA^\perp)-K.
    \label{eq:block-Z-distance}
\end{equation}
To prove this, suppose, for contradiction, that some $z\in\calC_0^\perp\setminus\calC_1^\perp$ satisfies $\wt(z)<\dist(\calA^\perp)-K$.
By the definitions of shortening and puncturing, $z$ satisfies
\begin{enumerate}
    \item for every $c\in\calC'$ with $c_S=\vb{0}$, $c_{\overline{S}}\cdot z=0$;
    \item there exists $g\in\calC'$ such that $g_{\overline{S}}\cdot z=1$.
\end{enumerate}
Write the vector in the second property as $g=(h,\sigma^{\oplus(n_{\calA}-K)}(u_{\overline{S_{\calA}}}))$, where $u\in\calA$, $h\in\FF_2^K\setminus\{0\}$, and $u_{S_{\calA}}=xh$.
Write $z=(z_P)_{P\in\overline{S_{\calA}}}$ in $L$-bit blocks and let $J_{\calA}\coloneqq\{P\in\overline{S_{\calA}}:z_P\neq0\}$ index its nonzero blocks.
Every nonzero block $z_P$ contains at least one nonzero entry, so $|J_{\calA}|\leq\wt(z)$.
Our assumption therefore gives $|S_{\calA}\cup J_{\calA}|=K+|J_{\calA}|<\dist(\calA^\perp)$.
Thus, the projection of $\calA$ onto $S_{\calA}\cup J_{\calA}$ is $\FF_{1024}^{|S_{\calA}\cup J_{\calA}|}$.
Consequently, we may choose $u'\in\calA$ with $u'_{S_{\calA}}=0$ and $u'_{J_{\calA}}=u_{J_{\calA}}$.
Its corresponding codeword in $\calC'$ has first $K$ coordinates zero.
By the first property and the agreement of $u'$ with $u$ on $J_{\calA}$,
\begin{equation*}
    0=z\cdot\sigma^{\oplus(n_{\calA}-K)}
       \left(u'_{\overline{S_{\calA}}}\right)
       =z\cdot\sigma^{\oplus(n_{\calA}-K)}
          \left(u_{\overline{S_{\calA}}}\right)=1,
\end{equation*}
where the last equality follows from the choice of $g$ in the second property.
This contradiction proves \cref{eq:block-Z-distance}.

It remains to choose $K$.
Set $K\coloneqq\lfloor n_{\calA}/992\rfloor$, so $K/n_{\calA}\to1/992$.
For sufficiently large $n_{\calA}$, this choice satisfies $1\leq K<\dist(\calA^\perp)$, as required for $S$ in \cref{lem:weighted-divisibility-transversal-T} for fully addressable transversality.
The AG-code bounds in \cref{eq:direct-outer-limits}, together with $N=L(n_{\calA}-K)$, give
\begin{equation*}
    \frac{K}{N}\longrightarrow\frac1{991L},
    \qquad
    \frac{D_{X}}{N} \geq \frac{925}{991L},
    \qquad
    \frac{D_Z}{N}\geq\frac1{991L}.
\end{equation*}
Therefore, the resulting CSS codes are asymptotically good and admit addressable transversal $T$ gates.
\end{proof}

\section{Minimum-length binary embeddings}
\label{sec:minimum-length-embeddings}

The binary embedding in Section~\ref{sec:binary-embedding} determines the physical block
length of the quantum code.
Finding the shortest embedding satisfying the fivefold identity amounts
to minimizing the Hamming weight of a solution to a linear system over $\FF_2$.
Explicit embeddings and the resulting quantum-code parameters are
given in \cref{app:explicit-parameters}.

We optimize the binary embedding length for general $q=2^m$ with $m\geq1$.
For the AG-code construction used here, $m=10$ is the smallest admissible extension degree; every even $m\geq10$ gives a $15$-multiplication family with positive asymptotic rate and primal and dual relative-distance bounds, as shown in \cref{appendix:required-ag-code-field}.
Let $\Phi_m:\FF_q^5\to\FF_2$ be defined by
\cref{eq:quintic-trace-form} using the absolute trace $\Tr_m: \FF_{2^m} \rightarrow \FF_2$.
The multilinearity and repeated-input identities proved in
Appendix~\ref{appendix:properties-of-binary-encoding} extend to $\Phi_m$,
since their proofs use only characteristic two and Frobenius invariance
of the absolute trace.
By \cref{eq:quintic-trace-form},
$\Phi_m(x,x,x,x,y)=\Tr_m(xy)$.
{Let $L_{\min}(m)$ denote the minimum length of a binary linear map
$\sigma:\FF_q\to\FF_2^L$ satisfying}
\begin{equation}
    \sum_{s=1}^L\prod_{a=1}^5\sigma_s(x_a)
       =\Phi_m(x_1,\ldots,x_5)
       \qquad(x_1,\ldots,x_5\in\FF_q).
    \label{eq:addressable-fivefold-embedding}
\end{equation}

{For a binary basis $e_1,\ldots,e_m$ of $\FF_q$, write
$x=\sum_{j=1}^m x_je_j$ and set $\ell_S(x)=\sum_{j\in S}x_j$ for
each nonempty $S\subseteq[m]$, as in
\cref{eq:binary-coordinate-functional}.}
{Using the notation of \cref{eq:pure-fifth-power-decomposition}, let
$\zeta\coloneqq(\zeta_S)_{\varnothing\ne S\subseteq[m]}\in\FF_2^{q-1}$
encode the coordinate choices, where $\zeta_S=0$ means that $\ell_S$
is omitted from $\sigma$.}
{We allow every nonempty $S\subseteq[m]$; the construction in
\cref{sec:binary-embedding} uses only sets with $|S|\leq5$.}

Let $s_m=\sum_{a=1}^{\min\{5,m\}}\binom ma$, and let
$M_m\in\FF_2^{s_m\times(q-1)}$ be the binary matrix whose rows are
indexed by nonempty subsets $I\subseteq[m]$ of size at most five and
whose columns are indexed by nonempty subsets $S\subseteq[m]$.
Its entries are $ (M_m)_{I,S}=\mathbf1_{I\subseteq S}.$
{Here $\mathbf1_{I\subseteq S}$ is $1$ if $I\subseteq S$ and $0$
otherwise.}
For each $I=\{i_1,\ldots,i_r\}$, define the $I$th entry of
$\gamma_m\in\FF_2^{s_m}$ by
\begin{equation*}
    \gamma_m(I)
       =\Phi_m(e_{i_1},\ldots,e_{i_r},
                    \underbrace{e_{i_r},\ldots,e_{i_r}}_{5-r}).
\end{equation*}
The repeated-input identities make this independent of the ordering
and the choice of repeated inputs.

\begin{lemma}
\label{lem:addressable-affine-minimum}
For $\zeta\in\FF_2^{q-1}$, let
$\sigma(x)=(\ell_S(x))_{\zeta_S=1}$. Then $\sigma$ satisfies
\cref{eq:addressable-fivefold-embedding} if and only if
$M_m\zeta=\gamma_m$.  Every map $\sigma$ satisfying the identity is injective, and
\begin{equation}
    L_{\min}(m)
       =\min\left\{\sum_{\varnothing\ne S\subseteq[m]}\zeta_S:
           \zeta\in\{0,1\}^{q-1},\ M_m\zeta=\gamma_m\right\}.
    \label{eq:addressable-affine-minimum}
\end{equation}
Moreover, $\operatorname{rank}M_m=s_m$ and
$m\leq L_{\min}(m)\leq s_m$.
\end{lemma}

\begin{proof}
Repeating the inputs in \cref{eq:addressable-fivefold-embedding}
gives $\sigma(x)\cdot\sigma(y)=\Tr_m(xy)$.
The trace pairing is nondegenerate, so $\sigma$ is injective.
Each coordinate is a binary linear functional.
{Deleting zero coordinate functionals and pairs of identical
coordinate functionals preserves \cref{eq:addressable-fivefold-embedding}.}
{A minimum-length map therefore has distinct nonzero coordinate
functionals $\ell_S$ from \cref{eq:binary-coordinate-functional}.}

For basis inputs $e_{i_1},\ldots,e_{i_5}$, let
$I=\{i_1,\ldots,i_5\}$. Since $\ell_S(e_i)=\mathbf1_{i\in S}$,
the product $\prod_{a=1}^5\ell_S(e_{i_a})$ is $1$ exactly when
$I\subseteq S$, and $0$ otherwise. The fivefold identity on these
inputs is therefore $\sum_{\{S\mid I\subseteq S\}}\zeta_S=\gamma_m(I)$
in $\FF_2$.
Both sides are binary multilinear, so these equations are also
sufficient for the identity on all inputs.  The number of selected
functionals is $\sum_S\zeta_S$, proving the optimization formula.

Restrict $M_m$ to columns with $1\leq|S|\leq5$, ordered with the
rows by decreasing cardinality.  The resulting square matrix is
triangular with diagonal entries one.  Hence $M_m$ has rank $s_m$,
and a solution supported on these $s_m$ columns exists for every
target $\gamma_m$.  This proves the upper bound on $L_{\min}(m)$;
injectivity proves the lower bound.
\end{proof}

{The recursion in \cref{eq:quintic-subset-recursion}, with
$\zeta_S=0$ for $|S|>5$, gives a particular solution
$\zeta^{(0)}$ of $M_m\zeta=\gamma_m$.}
{The solution set is $\zeta^{(0)}+\ker M_m$, with $2^{q-1-s_m}$
elements by \cref{lem:addressable-affine-minimum}.  Gaussian
elimination computes a basis of $\ker M_m$.}
{By \cref{eq:addressable-affine-minimum}, enumerating this affine
space and taking the least Hamming weight computes $L_{\min}(m)$.}

Suppose a family of AG codes over a fixed $\FF_q$ has
the $15$-multiplication property, positive rate, and positive relative
distances for both the codes and their duals. Any solution of
$M_m\zeta=\gamma_m$ defines a binary embedding satisfying
\cref{eq:addressable-fivefold-embedding} for $\Phi_m$ defined by
\cref{eq:quintic-trace-form}. The argument of
\cref{lem:good-weighted-divisible-code} then gives binary codes with
the $4$-multiplication property. The restriction and
weighted-divisibility steps of \cref{thm:direct-good-transversal-T}
produce an asymptotically good CSS-code family with fully addressable
transversal $T$ gates.

For the resulting CSS family, the physical length is
$N=L(n_{\calA}-K)$. With the AG code family and $K$ fixed, the stated
lower bounds on $D_X$ and $D_Z$ do not depend on $L$. Thus a shorter
embedding improves
the constants in the asymptotic rate and relative-distance bounds in
\cref{eq:direct-quantum-parameters}, each of which scales as $1/L$.


\section[Discussion and future directions]{Discussion and future directions}
\label{sec:discussion}

{We construct binary CSS codes with constant rate and linear
distance that admit fully addressable transversal logical $T$ gates on
a fixed encoding, without any subsequent Clifford correction.
AG-code multiplication and a binary embedding yield fivefold overlap
identities, while parity lifting supplies divisibility coefficients that
prescribe individual logical phases. Restriction and compression of
selected coordinates preserve these identities and yield linear lower
bounds on the quantum distances. We also formulate the minimum embedding
length as a binary optimization problem and obtain shorter embeddings
that improve the explicit parameter bounds.}

The construction raises questions about the coding and
computational cost of full addressability.
\begin{enumerate}
\item {What general structural principles govern fully
addressable transversal non-Clifford gates in asymptotically good quantum
codes? In particular, can the prescribed-coefficient mechanism be
connected with cohomological and tensor-network
constructions~\cite{golowich2026improved,cao2026quantum} to obtain criteria
that ensure both independent logical control and asymptotic goodness?
A common framework could reveal new code families and clarify which
features of the AG construction are essential.}

\item Do asymptotically good binary quantum LDPC codes admit fully
addressable transversal logical $T$ using physical powers of $T$, with a
fixed encoding and no subsequent Clifford correction? Here, the
stabilizer generators must have uniformly bounded weight, and each qubit
must occur in a uniformly bounded number of generators. Recent
constructions provide addressable controlled-phase gates with other
parameter tradeoffs~\cite{golowich2026improved}, and nontrivial transversal
multi-controlled-$Z$ gates on good quantum locally testable
codes~\cite{li2026goodtransversal}. These results do not establish the
fully addressable logical $T$ property required here.

\item What is the optimal space-time tradeoff for universal
fault-tolerant computation using codes with fully addressable transversal
$T$? Constant space overhead and logarithmic time overhead are already
achievable under suitable circuit-width and classical-processing
assumptions~\cite{han2026purely}. The question is
whether direct, parallel logical $T$ layers improve this tradeoff once
noisy syndrome extraction, ancillary-state preparation, complementary
logical gates, and classical processing are included under a common noise
and connectivity model.
\end{enumerate}

\paragraph{Statement of AI use.}
We used GPT 5.6 Sol and GPT 6 Astra to assist with brainstorming ideas, exploring proof strategies, writing proofs and preparing manuscript. 
The authors take full responsibility for every claim, proof, and citation in it.

\section*{Acknowledgements}

T.L. and B.W. acknowledge support
from the National Natural Science Foundation of China Grant No. 12405014.
B.C. acknowledges the support by the CQT Young Researcher Career Development Grant (26-YRCDG-BC).

\bibliographystyle{alpha}
\bibliography{refs}

\appendix
\crefalias{section}{appendix}
\crefname{appendix}{Appendix}{Appendices}
\Crefname{appendix}{Appendix}{Appendices}
\section{Algebraic geometry codes and their multiplication property}
\label{app:ag-codes}

Here, we recall the definitions and standard results on AG codes used in
\cref{subsec:ag-fifteen-multiplication}, and refer to
\cite{stichtenoth2009algebraic} for details.  We give proofs of the resulting code parameter bounds and the divisor
criterion for the multiplication property.

\subsection{From polynomial evaluation to rational places}

For a field $K$, the polynomial ring $K[x]$ consists of formal
expressions $f(x)=\sum_{j=0}^r c_jx^j$ with $c_j\in K$; equality
means equality of the coefficients.  Here, $x$ is a formal variable.
For nonzero $f$, its degree is the largest index $j$ with $c_j\neq0$,
and the coefficient at that index is its leading coefficient.
Evaluation at $\alpha\in K$ substitutes $\alpha$ for $x$, giving
$f(\alpha)=\sum_{j=0}^r c_j\alpha^j$.
For distinct $a_1,\ldots,a_n\in\FF_q$, evaluating polynomials of
degree at most $r$ gives the Reed-Solomon code
\begin{equation*}
    \{(f(a_1),\ldots,f(a_n)):f\in\FF_q[x],\ \deg f\leq r\},
    \qquad 0\leq r<n,
\end{equation*}
where the zero polynomial is included.  The message space has basis
$1,x,\ldots,x^r$.  Formal polynomials retain their coefficients even
when their evaluations agree: for example, $x^q-x$ is a nonzero
polynomial that vanishes at every element of $\FF_q$.

The \emph{rational function field} $K(x)$ consists of fractions
$a/b$ with $a,b\in K[x]$ and $b\neq0$, where
$a/b=c/d$ means $ad=bc$, with the usual addition and multiplication
of fractions.  A fraction is in \emph{lowest terms}
when numerator and denominator have no common factor of positive
degree.  For such a representation, evaluation at $x=\alpha$ is
$a(\alpha)/b(\alpha)$ when $b(\alpha)\neq0$; otherwise the function
has a \emph{pole} there.  Canceling common factors before evaluation
makes this rule independent of the representation.

To pass to more general function fields, recall that a \emph{field
extension} $E/K$ means that $E$ is a field containing $K$.
Multiplication by elements of $K$ makes $E$ a vector space over $K$.
Its \emph{extension degree} is $[E:K]=\dim_K E$, and the extension
is \emph{finite} when this dimension is finite.  An element $u\in E$
is \emph{algebraic over $K$} if $p(u)=0$ for some nonzero polynomial
$p(U)\in K[U]$.  It is \emph{transcendental over $K$} if no such
polynomial exists.  In particular, the formal variable $x\in K(x)$
is transcendental over $K$: substituting $x$ into a nonzero polynomial
$p(U)$ gives the nonzero formal polynomial $p(x)$.

A \emph{function field} $F/\FF_q$ is a field containing
$\FF_q(x)$ with $[F:\FF_q(x)]<\infty$.
Thus, its elements form a finite-dimensional vector space over the
field of rational functions.  The simplest example is
$F=\FF_q(x)$ itself.  Throughout, we assume that $\FF_q$ is the
\emph{full constant field}: any element of $F$ satisfying a nonzero
polynomial equation with coefficients in $\FF_q$ already belongs to
$\FF_q$.  Places provide evaluation points for these fields, and
valuations record the orders of zeros and poles at those points.

\begin{definition}[Places and evaluation]
\label{def:app-ag-evaluation}
A \emph{place} $P$ of $F/\FF_q$ is specified by a \emph{normalized
discrete valuation} $v_P:F\to\ZZ\cup\{\infty\}$.  Writing
$F^\times=F\setminus\{0\}$, this means that $v_P(0)=\infty$,
$v_P(F^\times)=\ZZ$, $v_P(c)=0$ for $c\in\FF_q^\times$, and
\begin{equation*}
    v_P(fh)=v_P(f)+v_P(h),\qquad
    v_P(f+h)\geq\min\{v_P(f),v_P(h)\}.
\end{equation*}
For $f\neq0$, positive and negative values of $v_P(f)$ describe a
zero and a pole, respectively, with order $|v_P(f)|$.
A function is \emph{regular at $P$} if $v_P(f)\geq0$.
The regular functions form the \emph{valuation ring}
$\calO_P$, and we also define the subset $\frakm_P$ by
\begin{equation*}
    \calO_P=\{f\in F:v_P(f)\geq0\},\qquad
    \frakm_P=\{f\in F:v_P(f)>0\}.
\end{equation*}
We identify two regular functions $f,h$ when their difference lies
in $\frakm_P$, that is, when $v_P(f-h)>0$.
Their equivalence classes form the \emph{residue field}
$F_P=\calO_P/\frakm_P$; the notation
$f+\frakm_P$ denotes the class of $f$.
This is a finite extension of $\FF_q$.  The \emph{degree} of $P$
is $\deg P=[F_P:\FF_q]$, and $P$ is \emph{rational} if
$\deg P=1$.  Evaluation of a regular function is
$f(P)=f+\frakm_P\in F_P$.
\end{definition}

Addition and multiplication of classes are induced by the operations
in $\calO_P$; the valuation inequalities make these operations
independent of the chosen representatives.  A nonzero class has a
representative $f$ with $v_P(f)=0$, so $v_P(f^{-1})=0$ and the class
of $f^{-1}$ is its inverse.  This explains why $F_P$ is a field.
Its finite extension degree is a standard function-field result
\cite[Proposition~1.1.15]{stichtenoth2009algebraic}.
Constants embed in $F_P$, since distinct constants have a difference
of valuation zero.  At a rational place this identifies $F_P$ with
$\FF_q$.  For regular $f,h$, evaluation therefore satisfies
\begin{equation*}
    (f+h)(P)=f(P)+h(P),\qquad (fh)(P)=f(P)h(P),
\end{equation*}
and $f(P)=0$ exactly when $v_P(f)>0$.

For $F=\FF_q(x)$, the places have a concrete description
\cite[Section~1.2]{stichtenoth2009algebraic}.
Let $h\in\FF_q[x]$ be \emph{monic}, meaning that its leading
coefficient is one, and \emph{irreducible}, meaning that it has
positive degree and is not a product of two polynomials of positive
degree.  It defines a place $P_h$ with
\begin{equation*}
    v_{P_h}(a/b)=\ord_h(a)-\ord_h(b),
    \qquad \deg P_h=\deg h,
\end{equation*}
for nonzero $a,b$, where $\ord_h$ is the multiplicity
of $h$ as a factor.  Its residue field $\FF_q[x]/(h)$ consists of
polynomial remainders of degree less than $\deg h$, with addition and
multiplication followed by reduction modulo $h$.
In particular, for $h=x-\alpha$ with $\alpha\in\FF_q$, the place
is rational and evaluation agrees with substitution at $\alpha$.
There is one further place, the rational place $P_\infty$, with
\begin{equation*}
    v_{P_\infty}(a/b)=\deg b-\deg a.
\end{equation*}
Thus, a nonconstant polynomial has a pole at infinity whose order is
its degree.  This expresses the degree bound in polynomial evaluation
as a bound on poles.

\subsection{Divisors and Riemann-Roch spaces}

A \emph{divisor} is a formal sum $G=\sum_Pc_PP$ over the places
of $F$, where $c_P\in\ZZ$ and only finitely many coefficients are
nonzero.  Its support is $\supp G=\{P:c_P\neq0\}$,
and its degree is $\deg G=\sum_Pc_P\deg P$.
We call $G$ \emph{effective}, and write $G\geq0$, if every
coefficient is nonnegative; $G\leq H$ means $H-G\geq0$.
Addition and these inequalities are understood coefficientwise.
For $f\in F^\times$, its \emph{principal divisor} is
$(f)=\sum_Pv_P(f)P$.  It has finite support and degree zero
\cite[Sections~1.3--1.4]{stichtenoth2009algebraic}, so the total
degrees of its zeros and poles agree, counting multiplicities.

The \emph{Riemann-Roch space} associated with $G$ is
\begin{equation}
    \calL(G)=\{0\}\cup\{f\in F^\times:(f)+G\geq0\},
    \qquad \ell(G)=\dim_{\FF_q}\calL(G).
    \label{eq:app-ag-rr-space}
\end{equation}
The valuation inequalities make $\calL(G)$ an $\FF_q$-vector space;
its finite dimensionality is a standard result
\cite[Proposition~1.4.9]{stichtenoth2009algebraic}.
A positive coefficient $c_P$ permits a pole of order at most $c_P$,
while a negative coefficient requires a zero of order at least $-c_P$.
For example, if $P\neq Q$, then $\calL(3P-2Q)$ permits a pole of
order at most three at $P$, requires a zero of order at least two at
$Q$, and permits no pole at any other place.
The valuation rules immediately give
\begin{equation}
    \begin{aligned}
        G\leq H&\quad\Longrightarrow\quad
            \calL(G)\subseteq\calL(H),\\
        f\in\calL(G),\ h\in\calL(H)&\quad\Longrightarrow\quad
            fh\in\calL(G+H).
    \end{aligned}
    \label{eq:app-ag-function-products}
\end{equation}
Also, $\deg G<0$ implies $\calL(G)=\{0\}$: a nonzero member would
make $(f)+G$ effective with negative degree.

Riemann's theorem states that the integers $\deg G+1-\ell(G)$,
as $G$ ranges over all divisors, have a finite upper bound
\cite[Section~1.4]{stichtenoth2009algebraic}.
Their maximum defines the \emph{genus}:
\begin{equation}
    g=\max_G\{\deg G+1-\ell(G)\}.
\end{equation}
Thus, $g$ bounds the loss of dimension relative to $\deg G+1$.
A function without poles is constant
\cite[Corollary~1.1.20]{stichtenoth2009algebraic}, so
$\calL(0)=\FF_q$ and $\ell(0)=1$.  Taking $G=0$ shows that
$g\geq0$.
We use the equivalent characterization that a divisor $R$ is
\emph{canonical} if and only if $\deg R=2g-2$ and $\ell(R)=g$
\cite[Corollary~1.5.16 and Proposition~1.6.2]{stichtenoth2009algebraic}.
The following standard theorem provides both its existence and the
dimension formula we use.

\begin{theorem}[Riemann-Roch]
\label{thm:app-ag-riemann-roch}
Canonical divisors exist.  For every canonical divisor $R$ and every
divisor $G$,
\begin{equation}
    \ell(G)=\deg G+1-g+\ell(R-G).
    \label{eq:app-ag-riemann-roch}
\end{equation}
Consequently, $\ell(G)\geq\deg G+1-g$, with equality whenever
$\deg G>2g-2$.
\end{theorem}

See \cite[Theorems~1.5.15 and 1.5.17]{stichtenoth2009algebraic} for a proof.
The equality in the last assertion follows because
$\deg(R-G)<0$ in the stated range.
For the rational function field $\FF_q(x)$,
\begin{equation*}
    \calL(rP_\infty)=\Span_{\FF_q}\{1,x,\ldots,x^r\}
    \qquad(r\geq0).
\end{equation*}
Indeed, a rational function without finite poles is a polynomial, and
its pole order at infinity is its degree.  Comparing the dimension
$r+1$ with Riemann-Roch for sufficiently large $r$ gives $g=0$.
Thus, $G$ generalizes the degree bound on the Reed-Solomon message
space, and $g$ measures the possible dimension loss in this
generalization.

\subsection{Evaluation codes and their parameters}

Choose distinct rational places $P_1,\ldots,P_n$, with $n\geq1$,
and set $D=P_1+\cdots+P_n$.  Let $G$ be a divisor whose support is
disjoint from that of $D$.  Every function in $\calL(G)$ is then
regular at each $P_i$, and rationality ensures that all values lie in
$\FF_q$.  Since every $P_i$ has degree one, $\deg D=n$.
The \emph{AG evaluation code} is the image
of the linear map
\begin{equation}
    \begin{aligned}
        \evmap_D:\calL(G)&\longrightarrow\FF_q^n,
        &f&\longmapsto(f(P_1),\ldots,f(P_n)),\\
        C_L(D,G)&\coloneqq\evmap_D(\calL(G)).
    \end{aligned}
    \label{eq:app-ag-evaluation-code}
\end{equation}

\begin{proposition}[Dimension and distance]
\label{prop:app-ag-parameters}
For $D,G$ as above,
\begin{equation}
    \ker(\evmap_D)=\calL(G-D),\qquad
    \dim_{\FF_q}C_L(D,G)=\ell(G)-\ell(G-D).
    \label{eq:app-ag-evaluation-kernel}
\end{equation}
If $\deg G<n$, evaluation is injective and
$\dim_{\FF_q}C_L(D,G)=\ell(G)\geq\deg G+1-g$.
In particular, if $2g-2<\deg G<n$, then
\begin{equation}
    \dim_{\FF_q}C_L(D,G)=\deg G+1-g.
    \label{eq:app-ag-code-dimension}
\end{equation}
Whenever $C_L(D,G)\neq\{0\}$, its minimum distance satisfies
\begin{equation}
    \dist(C_L(D,G))\geq n-\deg G.
    \label{eq:app-ag-code-distance}
\end{equation}
\end{proposition}

\begin{proof}
For $f\in\calL(G)$, vanishing at every $P_i$ is equivalent to
$v_{P_i}(f)\geq1$ for every $i$.  These are precisely the additional
conditions imposed by $G-D$, giving the kernel formula and hence
the dimension formula.  If $\deg G<n$, then $\deg(G-D)<0$, so
the kernel is zero.  The dimension bounds now follow from
\cref{thm:app-ag-riemann-roch}.

For a nonzero codeword represented by $f$, let $D_0$ be the sum of
the evaluation places where it vanishes.  If there are $z$ such
places, then $(f)+G-D_0\geq0$ and $\deg D_0=z$.  Taking degrees
gives $\deg G-z\geq0$, so the codeword has weight
$n-z\geq n-\deg G$.
\end{proof}

\subsection{Duality and the multiplication property}

The Euclidean dual of an AG evaluation code is again an evaluation
code for a suitable canonical divisor.  We use the following
standard form of AG-code duality without proof
\cite[Theorem~2.2.8, Lemma~2.2.9, and
Proposition~2.2.10]{stichtenoth2009algebraic}.

\begin{theorem}[AG-code duality]
\label{thm:app-ag-duality}
For $D=P_1+\cdots+P_n$ as above, there exists a canonical divisor
$R$, whose coefficient at every $P_i$ is $-1$, such that for every
divisor $G$ with support disjoint from $D$,
\begin{equation}
    C_L(D,G)^\perp=C_L(D,R+D-G).
    \label{eq:app-ag-duality}
\end{equation}
Here, the dual is with respect to the ordinary inner product
$\langle u,v\rangle=\sum_{i=1}^n u_iv_i$ over $\FF_q$.
\end{theorem}

In what follows, $R$ denotes a canonical divisor satisfying this
duality identity.  The coefficient condition ensures that $R+D-G$
has support disjoint from $D$, so its evaluation code is defined.
Whenever $C_L(D,G)^\perp\neq\{0\}$,
\cref{prop:app-ag-parameters} gives
\begin{equation}
    \dist(C_L(D,G)^\perp)
       \geq n-\deg(R+D-G)=\deg G+2-2g.
    \label{eq:app-ag-dual-distance}
\end{equation}
This is positive when $\deg G>2g-2$.

Recall that $\calA^{\star t}$ denotes the linear span of all
coordinatewise products of $t$ codewords from $\calA$, as in
\cref{def:star-product}.  Multiplicativity of evaluation relates
these products to multiplication of functions in $\calL(G)$.

\begin{proposition}[A divisor criterion for multiplication]
\label{prop:app-ag-multiplication}
Let $D,G,R$ be as above, and let $t\geq1$ be an integer.
If $(t+1)G\leq R+D$, then $\calA=C_L(D,G)$ has the
$t$-multiplication property:
\begin{equation}
    \calA^{\star t}\subseteq\calA^\perp.
    \label{eq:app-ag-multiplication}
\end{equation}
Equivalently, for every $u^{(1)},\ldots,u^{(t+1)}\in\calA$,
\begin{equation}
    \sum_{i=1}^n\prod_{a=1}^{t+1}u_i^{(a)}=0.
    \label{eq:app-ag-product-identity}
\end{equation}
\end{proposition}

\begin{proof}
Products of $t$ functions in $\calL(G)$ belong to $\calL(tG)$ by
\cref{eq:app-ag-function-products}.  Since evaluation preserves
products and $C_L(D,tG)$ is linear, it contains the span of the
corresponding component-wise products.
The assumed divisor inequality gives
$tG\leq R+D-G$, so inclusion of the function spaces and
\cref{thm:app-ag-duality} yield
\begin{equation*}
    \calA^{\star t}\subseteq C_L(D,tG)
       \subseteq C_L(D,R+D-G)=\calA^\perp.
\end{equation*}
The equivalent formulation follows by bilinearity of the inner
product and the definition of the star power.
\end{proof}

In \cref{lem:good-ag-fifteen-multiplication-code}, the
Galois tower supplies divisors over $\FF_{1024}$
satisfying $16G\leq R+D$.  The criterion with $t=15$ therefore
gives the required multiplication property.

\section{Multi-linearity and homogeneity of the $\Phi$ in binary embedding}
\label{appendix:properties-of-binary-encoding}
Here, we prove that the map $\Phi$ defined in
\cref{eq:quintic-trace-form} is multi-linear over $\FF_2$ and that
its value depends only on the set of distinct inputs, regardless of
their order or multiplicities.  We refer to the latter property as
homogeneity in this section.

Recall that $\Phi(x_1,\ldots,x_5)=\Tr(p(x_1,\ldots,x_5))$,
where
\begin{equation}
    p(x_1,\ldots,x_5)\coloneqq
       \prod_{\substack{S\subseteq[5]\\|S|\text{ odd}}}
          \left(\sum_{j\in S}x_j\right).
    \label{eq:app-phi-polynomial}
\end{equation}
All sums and products inside the trace are taken in $\FF_{1024}$, which has characteristic two.
In particular, identical terms cancel in pairs, and $(u+v)^{2^i}=u^{2^i}+v^{2^i}$ for every nonnegative integer $i$.
For all $u,v\in\FF_{1024}$ and $\zeta,\zeta'\in\FF_2$, the absolute trace satisfies
\begin{align}
    \Tr(\zeta u+\zeta' v)
       &=\zeta\Tr(u)+\zeta'\Tr(v),
       \label{eq:app-phi-trace-linearity}\\
    \Tr(u^{2^i})&=\Tr(u)
       \qquad(i\geq0).
       \label{eq:app-phi-trace-frobenius}
\end{align}
The first identity follows by expanding the trace, and the second
follows from $u^{2^{10}}=u$.

\emph{Multi-linearity.}
We prove that $\Phi$ is $\FF_2$-linear in $x_1$; by symmetry,
it is then $\FF_2$-linear in every input.  The polynomial inside
the trace can be written as
\begin{equation}
\begin{aligned}
    p(x_1,\ldots,x_5)
       &=x_2x_3x_4x_5 \cdot 
          \prod_{\substack{S\subseteq\{2,3,4,5\}\\|S|=3}}
             \left(\sum_{j\in S}x_j\right) \\
       &\qquad\cdot x_1\prod_{1<j<k\leq5}(x_1+x_j+x_k) \cdot
          (x_1+x_2+x_3+x_4+x_5)\\
       &=\gamma \cdot x_1 \cdot \prod_{1<j<k\leq5}(x_1+x_j+x_k) \cdot
          (x_1+x_2+x_3+x_4+x_5),
\end{aligned}
    \label{eq:app-phi-factorization}
\end{equation}
where $\gamma$ is the product of the factors without $x_1$.
Identical terms cancel in pairs in characteristic two.  Thus,
\begin{align}
    \Tr(p(x_1,\ldots,x_5))
       &=\begin{aligned}[t]
          &\Tr\!\Bigl[\gamma x_1(x_1+x_2+x_3)(x_1+x_2+x_4)(x_1+x_3+x_4)\\
          &\quad\cdot(x_1+x_2+x_5)(x_1+x_3+x_5)(x_1+x_4+x_5)(x_1+x_2+x_3+x_4+x_5)\Bigr]
          \end{aligned}\notag\\
       &=\begin{aligned}[t]
          &\Tr\!\Bigl[\gamma x_1(x_1+x_2+x_3)(x_1+x_2+x_4)(x_1+x_3+x_4)\\
          &\quad\cdot(x_1+x_2+x_5)\bigl((x_1+x_2+x_5)+x_2+x_3\bigr)\bigl((x_1+x_2+x_5)+x_2+x_4\bigr)\\
          &\quad\cdot\bigl((x_1+x_2+x_5)+x_3+x_4\bigr)\Bigr]
          \end{aligned}\notag\\
       &=\begin{aligned}[t]
          &\Tr\!\Bigl[\gamma(\chi_1x_1+\chi_2x_1^2+\chi_3x_1^3+\chi_4x_1^4)\\
          &\quad\cdot\Bigl(\chi_1(x_1+x_2+x_5)+\chi_2(x_1+x_2+x_5)^2+\chi_3(x_1+x_2+x_5)^3+\chi_4(x_1+x_2+x_5)^4\Bigr)\Bigr],
          \end{aligned}
    \label{eq:app-phi-quartic}
\end{align}
where
\begin{equation}
\begin{aligned}
    \chi_1&=(x_2+x_3)(x_2+x_4)(x_3+x_4),\\
    \chi_2&=x_2^2+x_3^2+x_4^2+x_2x_3+x_2x_4+x_3x_4,\\
    \chi_3&=0,\\
    \chi_4&=1.
\end{aligned}
    \label{eq:app-phi-quartic-coefficients}
\end{equation}
Since $\chi_3=0$, we have
\begin{align}
    \Tr(p(x_1,\ldots,x_5))
       &=\begin{aligned}[t]
          &\Tr\!\Bigl[\gamma(\chi_1x_1+\chi_2x_1^2+\chi_4x_1^4)\\
          &\quad\cdot\bigl(\chi_1(x_1+x_2+x_5)+\chi_2(x_1+x_2+x_5)^2+\chi_4(x_1+x_2+x_5)^4\bigr)\Bigr]
          \end{aligned}\notag\\
       &=\begin{aligned}[t]
          &\Tr\!\Bigl[\gamma(\chi_1x_1+\chi_2x_1^2+\chi_4x_1^4)\\
          &\qquad\cdot\Bigl((\chi_1x_1+\chi_2x_1^2+\chi_4x_1^4) +\chi_1(x_2+x_5)+\chi_2(x_2+x_5)^2
             +\chi_4(x_2+x_5)^4\Bigr)\Bigr]
          \end{aligned}\notag\\
       &=\begin{aligned}[t]
          &\Tr\!\Bigl[\gamma(\chi_1x_1+\chi_2x_1^2+\chi_4x_1^4)^2 + \\
          &\qquad\gamma(\chi_1x_1+\chi_2x_1^2+\chi_4x_1^4) \bigl(\chi_1(x_2+x_5)+\chi_2(x_2+x_5)^2+\chi_4(x_2+x_5)^4\bigr)\Bigr]
          \end{aligned}\notag\\
       &=\begin{aligned}[t]
          &\Tr\!\Bigl[\gamma\bigl(\chi_1^2x_1^2+\chi_2^2x_1^4+\chi_4^2x_1^8 + \\
          &\qquad\gamma(\chi_1x_1+\chi_2x_1^2+\chi_4x_1^4) \bigl(\chi_1(x_2+x_5)+\chi_2(x_2+x_5)^2+\chi_4(x_2+x_5)^4\bigr)\Bigr].
          \end{aligned}
    \label{eq:app-phi-expanded-powers}
\end{align}
Therefore, $p(x_1,\ldots,x_5)$ contains only the powers
$x_1,x_1^2,x_1^4,x_1^8$.  Rewrite its trace as
\begin{equation}
    \Tr(p(x_1,\ldots,x_5))
       =\Tr\!\bigl(
          \chi'_1x_1+\chi'_2x_1^2+\chi'_4x_1^4+\chi'_8x_1^8\bigr),
    \label{eq:app-phi-linearized-polynomial}
\end{equation}
where all factors involving $x_2,\ldots,x_5$ and the coefficients
have been absorbed into $\chi'_1,\chi'_2,\chi'_4,\chi'_8$.
These coefficients are independent of $x_1$.  For any
$\zeta,\zeta'\in\FF_2$, we have
\begin{align}
    &\Tr(p(\zeta x_1+\zeta' x_1',x_2,\ldots,x_5))\notag\\
    &\quad=\begin{aligned}[t]
       &\Tr\!\Bigl[\chi'_1(\zeta x_1+\zeta' x_1')
          +\chi'_2(\zeta x_1+\zeta' x_1')^2+\chi'_4(\zeta x_1+\zeta' x_1')^4
          +\chi'_8(\zeta x_1+\zeta' x_1')^8\Bigr]
       \end{aligned}\notag\\
    &\quad=\begin{aligned}[t]
       &\Tr\!\Bigl[\chi'_1(\zeta x_1+\zeta' x_1')
          +\chi'_2(\zeta x_1^2+\zeta' (x_1')^2)+\chi'_4(\zeta x_1^4+\zeta' (x_1')^4)
          +\chi'_8(\zeta x_1^8+\zeta' (x_1')^8)\Bigr]
       \end{aligned}\notag\\
    &\quad=\zeta\Tr\!\bigl(
       \chi'_1x_1+\chi'_2x_1^2+\chi'_4x_1^4+\chi'_8x_1^8\bigr)+\zeta'\Tr\!\bigl(
       \chi'_1x_1'+\chi'_2(x_1')^2+\chi'_4(x_1')^4+\chi'_8(x_1')^8\bigr)\notag\\
    &\quad=\zeta\Tr(p(x_1,x_2,\ldots,x_5))
       +\zeta'\Tr(p(x_1',x_2,\ldots,x_5)).
    \label{eq:app-phi-polynomial-linearity}
\end{align}
This proves that $\Phi$ is $\FF_2$-linear in $x_1$:
\begin{equation}
    \Phi(\zeta x_1+\zeta' x_1',x_2,\ldots,x_5)
       =\zeta\Phi(x_1,x_2,\ldots,x_5)
          +\zeta'\Phi(x_1',x_2,\ldots,x_5).
    \label{eq:app-phi-linearity}
\end{equation}
By symmetry, $\Phi$ is $\FF_2$-linear in every input.

\emph{Homogeneity.}
The value of $\Phi$ depends only on the set of distinct inputs,
regardless of their order or multiplicities.
For example, the
inputs $x,x,x,y,z$ and $y,y,x,x,z$ yield the same output.
Specifically, for four distinct inputs, homogeneity follows directly from the symmetry.
For the cases of three and two distinct inputs, we need to prove
\begin{align}
    \Phi(x,x,y,y,z)&=\Phi(x,y,y,y,z),
       \label{eq:app-phi-three-distinct}\\
    \Phi(x,y,y,y,y)&=\Phi(x,x,y,y,y)
       \label{eq:app-phi-two-distinct}
\end{align}
for all $x,y,z\in\FF_{1024}$.  

We first prove \cref{eq:app-phi-three-distinct}:
\begin{align}
    \Phi(x,x,y,y,z)
       &=\begin{aligned}[t]
          &\Tr\!\Bigl[x^2y^2z\cdot(2x+y)^2(x+2y)^2
             (2x+z)(2y+z)(x+y+z)^4(2x+2y+z)\Bigr]
          \end{aligned}\notag\\
       &=\Tr\!\bigl(
          x^2y^2z\cdot y^2x^2\cdot z^2\cdot(x^4+y^4+z^4)\cdot z\bigr)
          \notag\\
       &=\Tr\!\bigl(x^4y^4z^4(x^4+y^4+z^4)\bigr)
          \notag\\
       &=\Tr\!\bigl(xyz(x+y+z)\bigr),
          \label{eq:app-phi-xxyyz}\\
    \Phi(x,y,y,y,z)
       &=\begin{aligned}[t]
          &\Tr\!\Bigl[xy^3z\cdot(x+2y)^3(2y+z)^3
             (x+y+z)^3 \cdot 3y\cdot(x+3y+z)\Bigr]
          \end{aligned}\notag\\
       &=\Tr\!\bigl(
          xy^3z\cdot x^3z^3\cdot(x+y+z)^3\cdot y\cdot(x+y+z)\bigr)
          \notag\\
       &=\Tr\!\bigl(x^4y^4z^4(x+y+z)^4\bigr)
          \notag\\
       &=\Tr\!\bigl(xyz(x+y+z)\bigr)
          =\Phi(x,x,y,y,z).
          \label{eq:app-phi-xyyyz}
\end{align}
Setting $z=y$ in the above equations gives
\begin{align}
    \Phi(x,x,y,y,y)
       &=\Tr\!\bigl(xy^2(x+2y)\bigr)
          =\Tr(xy^2x)=\Tr(xy),
          \notag\\
    \Phi(x,y,y,y,y)
       &=\Tr\!\bigl(xy^2(x+2y)\bigr)
          =\Tr(xy)=\Phi(x,x,y,y,y).
          \label{eq:app-phi-two-input-calculation}
\end{align}
In addition, when all inputs are equal, 
\begin{equation*}
    \Phi(x,x,x,x,x) = \Tr(x^2) = \Tr(x),
\end{equation*}
which is used to prove the parity identity in \cref{eq:sigma-trace-identities} in the main text.

\section{Field size required by the AG-code construction}
\label{appendix:required-ag-code-field}

We show that $m=10$ is the smallest extension degree for which the
Galois tower, divisor criterion, and parameter bounds used
in \cref{lem:good-ag-fifteen-multiplication-code} give the
$15$-multiplication property together with positive asymptotic rate and
primal and dual relative-distance bounds of AG codes.  The argument follows the
tower and duality calculations of
\cite[Theorem~3.6 and Remark~3.8]{nguyen2025good}.

Write $q=\ell^2=2^m$, where $\ell=2^{m/2}$ is a power of $2$, so $m$
is even.  At tower level $i$, write $n_i=n_{\calA_i}$.  The tower has
$g_i=1+\frac{n_i}{\ell-1}(1-\frac1{\alpha_i}-\frac1{\beta_i})$,
where $\alpha_i,\beta_i\to\infty$.  Its positive divisors satisfy
$\alpha_i\deg A_i=\beta_i\deg B_i=n_i/(\ell-1)$, and its canonical
divisor is $R_i=(\ell\alpha_i-2)A_i+(\beta_i-2)B_i-D_i$, with
$\deg D_i=n_i$.

For $\calA_i=C_L(D_i,G_i)$, the criterion in
\cref{prop:app-ag-multiplication} and AG-code duality give the
$15$-multiplication property when $16G_i\leq R_i+D_i$.  Taking degrees,
using $\deg R_i=2g_i-2$, yields
$16(\deg G_i+2-2g_i)\leq n_i-30g_i+30$.
Since $g_i/n_i\to1/(\ell-1)$, a positive asymptotic relative
dual-distance bound from \cref{eq:app-ag-dual-distance} requires
$1-30/(\ell-1)>0$, or $\ell>31$.  As $\ell$ is a power of $2$, this
requires $\ell\geq32$ and $m\geq10$.

For these values, choose
\begin{equation*}
    G_i=
      \left\lfloor\frac{\ell\alpha_i-2}{16}\right\rfloor A_i
      +\left\lfloor\frac{\beta_i-2}{16}\right\rfloor B_i.
\end{equation*}
It satisfies $16G_i\leq R_i+D_i$.  The floor inequality
$x-1\leq\lfloor x\rfloor\leq x$ and the degree relations above give
$\deg G_i/n_i\to(\ell+1)/(16(\ell-1))$.  For $\ell>31$, this limit
exceeds $(2g_i-1)/n_i\to2/(\ell-1)$ and is less than $1$.  Thus
$2g_i-1\leq\deg G_i<n_i$ for all sufficiently large $i$.  The
AG-code parameter bounds in
\cite[Lemmas~3.3 and 3.5]{nguyen2025good} now give
\begin{align*}
    \frac{k_{\calA_i}}{n_i}
       &\longrightarrow\frac{\ell-15}{16(\ell-1)},\\
    \liminf_{i\to\infty}\frac{\dist(\calA_i)}{n_i}
       &\geq\frac{15\ell-17}{16(\ell-1)},\\
    \liminf_{i\to\infty}\frac{\dist(\calA_i^\perp)}{n_i}
       &\geq\frac{\ell-31}{16(\ell-1)}.
\end{align*}
All three are positive for every even $m\geq10$.  At the smallest
choice $m=10$, we have $\ell=32$ and
$\lfloor(\ell\alpha_i-2)/16\rfloor=2\alpha_i-1$, recovering the
divisor and the three bounds $17/496$, $463/496$, and $1/496$ in
\cref{lem:good-ag-fifteen-multiplication-code}.

\section{Explicit embeddings and numerical parameters}
\label{app:explicit-parameters}

{We specify an embedding of $\FF_{1024}$ satisfying
\cref{eq:addressable-fivefold-embedding} and derive the corresponding
quantum-code parameters from \cref{thm:direct-good-transversal-T}.}

{Set $m=10$ and $q=2^m=1024$.
Represent $\FF_q$ as $\FF_2[z]/(f_{10}(z))$, with basis
$e_j=\alpha^{j-1}$ for $1\leq j\leq10$, where $\alpha$ is the
residue class of $z$ and $f_{10}$ is the irreducible polynomial}
\begin{equation}
\begin{aligned}
 f_{10}(z)&=z^{10}+z^3+1.
\end{aligned}
\label{eq:embedding-field-polynomials}
\end{equation}
For $\varnothing\ne S\subseteq[m]$, put
$u(S)\coloneqq\sum_{j\in S}2^{j-1}$.  A coefficient vector
$\zeta$ in \cref{eq:addressable-affine-minimum} is represented by
the nonnegative integer
\begin{equation}
 W_q\coloneqq\sum_{\varnothing\ne S\subseteq[m]}
                    \zeta_S2^{u(S)-1},\qquad q=2^m,
 \label{eq:embedding-support-integer}
\end{equation}
where the coefficients are represented by $0$ and $1$.
Thus $\zeta_S$ is bit $u(S)-1$ of $W_q$, with the least significant
bit numbered zero.

{The hexadecimal representation of $W_{1024}$ is obtained
by concatenating the following four rows in order:}
\begin{equation}
\begin{gathered}
\texttt{001082408149D112000308640140880A0E4019468224442CA010140004C90815}\\
\texttt{830A90000215C01A41E232510282034018843451030840029C406AA400882102}\\
\texttt{452100C4306429128580C20095003022002809804A0800D1090C40D024804590}\\
\texttt{8008A56040020102104111050121405003400148102481306161139181280480}.
\end{gathered}
    \label{eq:addressable-244-mask}
\end{equation}

{The triangular construction in
\cref{eq:quintic-subset-recursion} gives a feasible vector of weight
$353$.  A randomized local search in the affine solution space
of Section~\ref{sec:minimum-length-embeddings} applies updates
$\zeta\mapsto\zeta+b$ with $b\in\ker M_{10}$, accepting all weight
reductions and equal-weight updates with probability $1/8$.
The resulting vector, encoded in \cref{eq:addressable-244-mask}, has
weight $244$ and satisfies $M_{10}\zeta=\gamma_{10}$, which certifies
$L_{\min}(10)\leq244$ by \cref{lem:addressable-affine-minimum}.}

{The rank formula in \cref{lem:addressable-affine-minimum}
gives $\dim_{\FF_2}\ker M_{10}=1023-s_{10}=386$.
For this coefficient vector, evaluating
$d_{\mathrm{in}}\coloneqq\min_{x\in\FF_{1024}^\times}\wt(\sigma(x))$
on all $1023$ nonzero field elements gives $d_{\mathrm{in}}=100$.
Thus the image of the specified embedding is a binary $[244,10,100]$
code.}

\begin{corollary}
\label{cor:multiplication-family-parameters}
{With the embedding specified by \cref{eq:addressable-244-mask},
the family in \cref{thm:direct-good-transversal-T} admits fully
addressable transversal $T$ gates and satisfies}
\begin{equation}
    \frac{K_i}{N_i}\longrightarrow\frac1{241804},\qquad
    \liminf_i\frac{D_{X,i}}{N_i}\geq\frac{92500}{241804},\qquad
    \liminf_i\frac{D_{Z,i}}{N_i}\geq\frac1{241804}.
    \label{eq:addressable-short-asymptotic-constants}
\end{equation}
Consequently, $\liminf_i D_i/N_i\geq1/241804$.
\end{corollary}

\begin{proof}
Apply
\cref{eq:direct-quantum-parameters}
with $(L,d_{\mathrm{in}})=(244,100)$, and use
$D_i=\min\{D_{X,i},D_{Z,i}\}$.
The claimed logical action follows from \cref{thm:direct-good-transversal-T}.
\end{proof}

\end{document}